\documentclass[12pt, centerh1]{article}

\usepackage{graphicx,lineno}
\usepackage{natbib}
\usepackage{amsmath, amsthm, amssymb, amsfonts}
\usepackage{supertabular,colonequals}
\usepackage{longtable, multirow}
\usepackage{rotating, eurosym}
\usepackage{subfigure}
\usepackage{epsfig}
\usepackage{color}
\modulolinenumbers[5]
\newtheorem{thm}{Theorem}

\newcommand{\load}{\mathbf\Lambda}

\newcommand{\ident}{\mathbf{I}}
\newcommand{\vecx}{\mathbf{x}}

\newcommand{\matc}{\mathbf{C}}
\newcommand{\vecb}{\mathbf{b}}

\newcommand{\vecX}{\mathbf{X}}
\newcommand{\vecE}{\mathbf{E}}
\newcommand{\vecd}{\mathbf{d}}
\newcommand{\veca}{\mathbf{a}}
\newcommand{\vecM}{\mathbf{M}}
\newcommand{\vecH}{\mathbf{H}}
\newcommand{\vecD}{\mathbf{D}}
\newcommand{\vecY}{\mathbf{Y}}

\newcommand{\vecU}{\mathbf{U}}
\newcommand{\vecK}{\mathbf{K}}

\newcommand{\vecr}{\mathbf{r}}

\newcommand{\vecS}{\mathbf{S}}
\newcommand{\vecz}{\mathbf{z}}
\newcommand{\vecy}{\mathbf{y}}
\newcommand{\vecZ}{\mathbf{Z}}
\newcommand{\vecu}{\mathbf{u}}

\newcommand{\zero}{\mathbf{0}}
\newcommand{\veclambda}{\mbox{\boldmath$\lambda$}}

\newcommand{\new}{\mbox{\tiny new}}

\newcommand{\vecmu}{\mbox{\boldmath$\mu$}}
\newcommand{\vecLambda}{\mbox{\boldmath$\Lambda$}}
\newcommand{\vecSigma}{\mbox{\boldmath$\Sigma$}}

\newcommand{\vecPhi}{\mbox{\boldmath$\Phi$}}

\newcommand{\vecTheta}{\mbox{\boldmath$\Theta$}}
\newcommand{\vecDelta}{\mbox{\boldmath$\Delta$}}
\newcommand{\vecOmega}{\mbox{\boldmath$\Omega$}}
\newcommand{\vectheta}{\mbox{\boldmath$\theta$}}

\newcommand{\vecepsilon}{\mbox{\boldmath$\epsilon$}}
\newcommand{\vecdelta}{\delta}
\newcommand{\vecalpha}{\mbox{\boldmath$\alpha$}}
\newcommand{\vecvarthet}{\mbox{\boldmath$\vartheta$}}

\newcommand{\matsig}{\mathbf{\Sigma}}

\newtheorem{cor}{Corollary}
\newtheorem{prop}{Proposition}
\newtheorem{defn}{Definition}

\title{Hidden Truncation Hyperbolic Distributions, Finite Mixtures Thereof, and Their Application for Clustering}  
  \author{Paula M.\ Murray$^{*}$, Ryan P.\ Browne$^{**}$ and Paul D.\ McNicholas$^{*}$}
  \date{\small $^{*}$Department of Mathematics and Statistics, McMaster University, Ontario, Canada.\\
 $^{**}$Department of Statistics and Actuarial Sciences, University of Waterloo, Ontario, Canada.}

\begin{document}

\maketitle 

\begin{abstract}
A hidden truncation hyperbolic (HTH) distribution is introduced and finite mixtures thereof are applied for clustering. A stochastic representation of the HTH distribution is given and a density is derived. A hierarchical representation is described, which aids in parameter estimation. Finite mixtures of HTH distributions are presented and their identifiability is proved. The convexity of the HTH distribution is discussed, which is important in clustering applications, and some theoretical results in this direction are presented. The relationship between the HTH distribution and other skewed distributions in the literature is discussed. Illustrations are provided --- both of the HTH distribution and application of finite mixtures thereof for clustering.\\

\end{abstract}

\section{Introduction}

Broadly, cluster analysis is the process of identifying groups of similar observations within a data set. Model-based clustering is a method for performing cluster analysis that involves the fitting of a finite mixture model. Recently, attention has focused on mixtures of flexible asymmetric distributions and, to this end, we develop a hidden truncation hyperbolic (HTH) distribution and use a mixture thereof to perform flexible model-based clustering. This mixture model is more general than other mixture models based on truncated distributions that have appeared in the literature. A finite mixture model has density 
\begin{equation}
f(\vecx~|~\vecvarthet)=\sum_{g=1}^{G}\pi_g f_g(\vecx~|~\vectheta_g),
\label{eq:generalmixture1}
\end{equation}
such that 
$\vecvarthet=(\pi_1,...,\pi_G,\vectheta_1,...,\vectheta_G)$, $\pi_g>0$, 
and $\sum_{g=1}^{G}\pi_g=1,$ where $\vecx$ is a $p$-dimensional data vector, $\pi_g$ is the $g$th mixing proportion, and $f_g(\vecx~|~\vectheta_g)$ is the $g$th component density.
%Mixture models have been studied in detail by \cite{everitt81}, \cite{titterington85}, \cite{mclachlan88}, and \cite{mclachlan00c}. 
%

Mixture models have been employed for decades in cluster analysis \citep[e.g.,][]{day69,wolfe63,wolfe65} and, in fact, \cite{mcnicholas16a} traces the association of mixture models with clustering back to \cite{tiedeman55}. The Gaussian mixture model has been the most popular approach and work in this direction has focused on approaches that reduce the number of free parameters in the component covariance matrices \citep[e.g.,][]{banfield93,celeux95,mcnicholas08}. 
%reparameterize the covariance matrix in a Gaussian mixture model using an eigenvalue decomposition.  They introduce parsimony by forcing some elements of this decomposition, representing cluster volume, shape, and orientation, respectively, to be the same for each cluster. \cite{celeux95} generalize this work to allow the volume, shape, and orientation of the clusters to be the same or different for all clusters. They further consider diagonal covariance matrices as well as spherical shapes for the clusters, resulting in 14 Gaussian parsimonious mixture models. \cite{fraley99} then introduce the {\tt mclust} package for {\sf R} \citep{R13}, implementing six of these 14 models as well as two additional models (subsequent versions of {\tt mclust} contain more models). %The most recent version of {\tt mclust} implements all 14 models \citep{mclust}. %
%Around this time, \cite{ghahramani97} introduce a mixture of factor analyzers (MFA) model and \cite{tipping99} introduce a mixture of probabilistic principal component analyzers (MPPCA). \cite{mclachlan00} further generalize the MFA model of \cite{ghahramani97}. \cite{mcnicholas08} build on the MPPCA and MFA models to develop a family of eight mixture models, and \cite{baek10} build on the MFA model to develop a mixture of common factor analyzers model. %These models reduce the number of parameters to be estimated in the component covariance matrices and allowed Gaussian mixture models to be applied to higher dimensional data. 
%
Over the past fifteen years or so, these Gaussian mixture modelling approaches have been extended in various ways to accommodate increasingly complicated applications. Many of these extensions involve using non-Gaussian distributions to accommodate component skewness and/or concentration \citep[e.g.,][]{franczak14,lee14,lin07,smcnicholas17,montanari10,murray14a}. %The HTH distribution encapsulates several of these non-Gaussian distributions. 
Interestingly, beyond certain formulations of the skew-normal and skew-$t$ distributions \citep[e.g.,][]{lee14,lin09,lin10}, models based on truncated distributions have remained relatively unexplored in the clustering literature. A recent review of model-based clustering is given by \cite{mcnicholas16b}.

The mixture of HTH distributions developed herein is based on a truncated hyperbolic random variable and contains both mixtures of certain formulations of mixtures of skew-$t$ and mixtures of skew-normal distributions as special cases. 
Some background on hidden truncation is provided in Section~\ref{sec:back} and the HTH distribution is developed in Section \ref{sec:HTH}, along with a derivation of the moments of the truncated symmetric hyperbolic distribution. Then, identifiability of finite mixtures of HTH distributions is proved (Section~\ref{sec:ident}) and convexity is discussed (Section~\ref{sec:concave}). An expectation-conditional maximization (ECM) algorithm is used for fitting mixtures of HTH distributions (Section~\ref{sec:param}), and clustering performance is illustrated on both real and simulated data (Section~\ref{sec:results}).

\section{Background}\label{sec:back}

%\subsection{Overview}
%Hidden truncation models describe observed variables that are truncated with respect to some hidden variables \citep[see][]{arnold93,arnold00,arnold02}. The hidden truncation hyperbolic distribution is developed using a hidden truncated random variable and is a ``super distribution" in two ways. 
%{\color{black}
%In this section, the hidden truncation hyperbolic distribution is developed. Begin by considering the symmetric hyperbolic distribution, which is a special case of \eqref{eq:ghddensity} with zero skewness (i.e.,\ $\vecalpha=\mathbf{0}$). In terms of background, we need to consider the generalized hyperbolic distribution (GHD), the generalized inverse Gaussian (GIG) distribution and one formulation of the skew-normal distribution.} 
%Note that we will introduce skewness into the HTH distribution in an alternative fashion (Section~\ref{sec:stochrep}).

\subsection{Generalized Inverse Gaussian Distribution}\label{sec:gig}
To formulate the HTH distribution, introduce a latent variable $W$ following the generalized inverse Gaussian (GIG) distribution. The random variable $W$ has density
$${\color{black}f_{\text{GIG}}}(w\mid\psi,\chi,\lambda)=\frac{(\psi/\chi)^{\lambda/2}w^{\lambda-1}}{2K_\lambda(\sqrt{\psi\chi})}\exp{\bigg(-\frac{\psi w + \chi/w}{2}\bigg)},$$
 for $w>0$ with $(\psi, \chi) \in \mathbb{R}_{+}^2$ and $\lambda\in\mathbb{R}$. We write $W\sim\text{GIG}(\psi,\chi,\lambda)$ to indicate that a random variable $W$ follows a GIG distribution.

\subsection{Generalized Hyperbolic Distribution}\label{sec:ghd}

The generalized hyperbolic distribution \citep[GHD;][]{mcneil05} has density
\begin{equation}
\begin{split}
f_{\text{GHD}}(\vecx\mid\vectheta)=
&\bigg\{\frac{\chi+(\vecx-\vecmu)^{\top}\vecSigma^{-1}(\vecx-\vecmu)}{\psi+\vecalpha^{\top}\vecSigma^{-1}\vecalpha}\bigg\}^{(\lambda-p/2)/2}
\\&\quad\times\frac{(\psi/\chi)^{\lambda/2}K_{\lambda-p/2}\bigg(\sqrt{(\psi+\vecalpha\vecSigma^{-1}\vecalpha)\big\{\chi+(\vecx-\vecmu)^{\top}\vecSigma^{-1}(\vecx-\vecmu)\big\}}\bigg)}{(2\pi)^{p/2}\mid\mathbf\Sigma\mid^{1/2}K_{\lambda}(\sqrt{\chi\psi}) \; \mbox{exp} \left\{  - (\vecx-\vecmu)^{\top}\vecSigma^{-1}\vecalpha \right\} },
\label{eq:ghddensity}
\end{split}
\end{equation}
where $\vecx$ is a $p$-dimensional data vector, $K_\lambda(\cdot)$ is the modified Bessel function of the third kind with index $\lambda$, and  $\vectheta=(\vecmu, \vecSigma, \vecalpha,\lambda, \chi, \psi)$ is a vector of parameters. This density function contains a $p$-dimensional skewness parameter $\vecalpha$. The GHD is mean-variance mixture and has a variety of limiting and special cases \citep[see][]{mcneil05}.
Setting $\lambda =1$ gives a multivariate generalized hyperbolic distribution such that its univariate margins are one-dimensional hyperbolic distributions. Fixing $\lambda = (p+1)/2$ gives a $p$-dimensional hyperbolic distribution, and $\lambda = -{1}/{2}$ gives the inverse-Gaussian distribution. 
For $\lambda >0$ and $\chi \rightarrow 0$, we have a limiting case known as the variance-gamma distribution \citep{barndorff78}. If $\lambda =1 $, $\psi= 2$ and $\chi\rightarrow 0$, we get the asymmetric Laplace distribution \citep[see][]{kotz2001} and if $\vecalpha = \mathbf{0}$, we have the symmetric hyperbolic distribution \citep{barndorff78}. Other special and limiting cases include the multivariate normal-inverse Gaussian (MNIG)  \citep{karlis07}, a skew-$t$  \citep{demarta2005}, the multivariate-$t$ and the Gaussian distribution. 

\cite{barndorff78} introduced generalized hyperbolic distributions to model the grain sizes of sand and diamonds. \cite{eberlein1995} used the hyperbolic distribution to model returns of German equities. \citep{browne13, tortora15,morris16} applied the GHD to the context of  model-based clustering.

The GHD given in \eqref{eq:ghddensity} is not identifiable because for any $\eta >0 $ we have that the parameters  $\vectheta=(\vecmu, \vecSigma, \vecalpha,\lambda, \chi, \psi)$  and $\vectheta^*=(\vecmu, \eta\vecSigma, \eta\vecalpha,\lambda, \chi/\eta, \eta\psi)$ yield the same density. 
% \ref{eq:ghddensity2} 
To ensure identifiability, we follow \cite{browne13} and set $\chi=\omega$ and $\psi=\omega$ such that  $W\sim\text{GIG}(\omega,\omega,\lambda)$. Using this parameterization and setting the skewness to zero (i.e.,\ $\vecalpha=\mathbf{0}$) we obtain an identifiable symmetric hyperbolic distribution with density
\begin{equation}
h(\vecx\mid\vecmu,\vecSigma,\lambda,\omega,\omega)= \bigg\{\frac{\omega+\delta(\vecx,\vecmu|\vecSigma)}{\omega}\bigg\}^{(\lambda-p/2)/2} \frac{K_{\lambda-p/2}\big(\sqrt{\omega\{\omega+\delta(\vecx,\vecmu|\vecSigma)\}}\big)}{(2\pi)^{p/2}\mid\mathbf\Sigma\mid^{1/2}K_{\lambda}(\omega)}.
\label{eq:ghddensity2}
\end{equation}
where $\delta(\vecx|\vecmu,\vecSigma)=(\vecx-\vecmu)^{\top}\vecSigma^{-1}(\vecx-\vecmu)$. The symmetric hyperbolic distribution has stochastic representation $$\vecX = \vecmu + \sqrt{W} \vecSigma^{1/2}  \vecK,$$  where $W \sim$ GIG$(\omega, \omega, \lambda)$ and { \color{black}$\vecK \sim\mathcal{N}_p(\mathbf{0}, \ident_p)$.}

\subsection{The Skew-Normal Distribution}

In the past several years, much work has been done on mixture modeling using asymmetric distributions. One form of the skew-normal distribution that has appeared often in this literature is the well-known skew-normal distribution of \cite{sahu03}  which has the characterization
 \begin{equation}
\vecX=\vecmu+\vecLambda|\vecZ_0|+\vecZ_1,
\label{sahu}
\end{equation}
where $|\vecZ| = ( | Z_1 |, \ldots, |Z_p|)$ for $\vecZ \in \mathbb{R}^p $,  $\vecX$ is a $p$-variate skew-normal random variable, $\vecLambda$ is a diagonal matrix that plays the role of a skewness parameter, and 
$$\left[ \begin{array}{c} \vecZ_0 \\ \vecZ_1 \end{array} \right] \sim \mathcal{N}_{2p} \left( \left[ \begin{array}{c} \mathbf{0}_p \\ \mathbf{0}_p \end{array} \right],\begin{bmatrix} \mathbf{I}_p & \mathbf{0}_p \\ \mathbf{0}_p & \vecSigma-\vecLambda\vecLambda^{\top} \end{bmatrix} \right).$$
{\color{black} Note that $|\vecZ_0|$ is a half-normal random variable, i.e., $|\vecZ_0| =  \vecU \sim \text{HN}\left(\ident_q\right)$.}  
 \cite{lee13c} refer to this as an ``unrestricted" skew-normal distribution because $\vecZ_0$ in the stochastic {\color{black}representation} \eqref{sahu} is a random vector rather than a scalar. This is in contrast to what they call a ``restricted" skew-normal distribution, like that of \cite{pyne09}, which has the characterization 
 \begin{equation}
\vecX=\vecmu+\veclambda| Z_0|+\vecZ_1,
\label{pyne}
\end{equation}
where
$$\left[ \begin{array}{c} Z_0 \\ \vecZ_1 \end{array} \right] \sim \mathcal{N}_{1+p} \left( \left[ \begin{array}{c} 0 \\ \mathbf{0}_p \end{array} \right],\begin{bmatrix} 1 & \mathbf{0}^{\top}_p \\ \mathbf{0}_p & \vecSigma -\veclambda\veclambda^{\top}\end{bmatrix} \right).$$
 Note that in the stochastic {\color{black}representation} in \eqref{pyne}, the random variable $Z_0$ is a scalar and the skewing parameter $\veclambda$ is a vector. 
 
 \cite{azzalini14} discuss why it is preferable to avoid using the terms ``restricted" and ``unrestricted" in this way. As such, we will refer to the skew-normal distribution of \cite{sahu03} as the SDB distribution herein, referring to the authors' initials, and we will refer to the skew-normal distribution of \cite{pyne09} as the classical skew-normal distribution.
Models based on the classical formulation are less computationally challenging to fit than their SDB counterparts. However, \cite{lee14} argue the SDB distribution is better at modelling asymmetric data.  The HTH distribution contains special cases that we will call HTHu ($q=1$) and HTHm ($1<q\leq p$) which contain the classical and SDB distributions, respectively, as special cases. Here ``u" in HTHu refers to the univariate integral in the HTH density \eqref{HTHdens} while the ``m" in HTHm refers to a multivariate integral (see Section~\ref{sec:HTH}). 

Herein a hyperbolic random variable is obtained via a relationship with a skew-normal random variable $\vecY$. A slightly more general formulation of the SDB skew-normal distribution is the canonical fundamental skew-normal (CFUSN) distribution \citep{arellano05} which has density
\begin{equation} 
f_{\tiny{\text{SN}}}(\vecy\mid\vecmu,\vecSigma,\vecLambda)=2^q\phi_p(\vecy\mid\vecmu,\vecOmega)\Phi_q\left\{\vecLambda^{\top}\vecOmega^{-1}(\vecy-\vecmu)\mid\vecDelta\right\},
\label{eq:sndensity}
\end{equation}
with location vector $\vecmu$, scale matrix $\vecSigma$, and $p\times q$ skewness matrix $\vecLambda$ such that $ \| \vecLambda^{\top} \vecz \| < 1$ for any unitary vector $z\in \mathbb{R}^p$, where $\vecOmega=\vecSigma+\vecLambda\vecLambda^{\top}$, $\vecDelta=\mathbf{I}_q-\vecLambda^{\top}\vecOmega^{-1}\vecLambda$ and $\phi_p(\cdot\mid\vecmu,\vecSigma)$ 
and $\vecPhi_q(\cdot\mid \vecSigma) \colonequals \vecPhi_q(\cdot\mid\vecmu = \zero_q,\vecSigma)$ are the density and cumulative distribution function of the multivariate normal distribution, respectively. 
The CFUSN density \eqref{eq:sndensity} is an  special case of the ``closed skew-normal'' (also called SUN) distribution, and \cite{arellano2006} gives an overview and give some the properties of the SUN distribution.

% the overall structure of (5) is the one of a special case of the so-called ?closed skew-normal? (also called SUN) distribution, widely studied in the recent distribution-theory literature; an overview is given by Azzalini & Arellano-Valle (2006). This connection should be mentioned.

Note that the SDB skew-normal distribution uses a less general form of this density with $q=p$ \citep{sahu03}. \cite{arellano07} show we can obtain a $p$-dimensional skew-normal random variable $\vecY$ by 
$$\vecY=\vecmu+\vecLambda\vecU+\vecSigma^{1/2}\vecK,$$
where $\vecU\sim \text{HN}(\mathbf{I}_q)$, $\vecK\sim\mathcal{N}_p(\mathbf{0},\mathbf{I}_p)$, and $\vecSigma^{1/2}$ denotes the square root of the matrix $\vecSigma$. Note that $\text{HN}(\mathbf{I}_q)$ denotes the half-normal distribution with scale matrix $\mathbf{I}_q$. We write $\vecY\sim\text{SN}_p(\vecmu,\vecSigma,\vecLambda)$ to denote $\vecY$ is a skew-normal random variable with density  \eqref{eq:sndensity}. This variant of the skew-normal distribution will be used herein to obtain a HTH random variable $\vecX$. {\color{black}\cite{arslan2015} introduce a mean-variance mixture using a GIG distribution and a skew-normal distribution. The resulting model includes a GHD and a skew-normal distribution as special cases.} %The HTH model {\color{black}is} a special case of this model if $q=1$.

\section{The Hidden Truncation Hyperbolic Distribution}\label{sec:HTH}

\subsection{Overview}
Hidden truncation models describe observed variables that are truncated with respect to some hidden variables; see \cite{arnold00,arnold02,arnold93}. The hidden truncation hyperbolic distribution is developed using a hidden truncated random variable and is a ``super distribution" as discussed in Section~\ref{sec:special}. We begin by considering the symmetric hyperbolic distribution, which is a special case of \eqref{eq:ghddensity} with zero skewness, i.e., $\vecalpha=\mathbf{0}$. The hidden truncation hyperbolic (HTH) distribution introduced herein regulates skewness by means of a $p\times q$ skewness matrix $\vecLambda$, where $1\leq q\leq p$. Therefore, we can fit special cases of our distribution by changing the value of~$q$.

\subsection{Stochastic Representation}\label{sec:stochrep}

 We generate a $p$-dimensional random variable $\vecX$ following the HTH distribution through the stochastic representation 
$\vecX=\vecmu+\sqrt{W}\vecY$,
 where $\vecY\sim\text{SN}_p(\mathbf{0},\vecSigma,\vecLambda)$ and $W\sim\text{GIG}(\omega,\omega,\lambda)$. The HTH density is
 \begin{equation}
 \begin{split}
&f_{\text{HTH}}(\vecx|\vectheta)=\\
&\ 2^qh_p(\vecx|\vecmu,\vecOmega,\lambda,\omega,\omega) H_q\left[\vecLambda^{\top}\vecOmega^{-1}(\vecx-\vecmu)\left\{\frac{\omega}{\omega+\vecdelta(\vecx|\vecmu,\vecOmega)}\right\}^{1/4}\bigg|\mathbf{0},\mathbf\vecDelta,\lambda-(p/2),\gamma,\gamma\right],
\label{HTHdens}
\end{split}
 \end{equation}
 where $\vectheta=(\vecmu,\vecSigma,\load,\lambda,\omega)$, $\gamma=\omega \sqrt{ 1+\delta(\vecx,\vecmu|\vecOmega)/\omega}$, $h_p(\cdot|\vecmu,\vecSigma,\lambda,\omega,\omega)$ is the density of a $p$-dimensional symmetric hyperbolic random variable and $H_q(\cdot|\vecmu,\vecSigma,\lambda,\omega,\omega)$ is the corresponding $q$-dimensional cdf.  Notice that skewness is introduced into this distribution through the $p\times q$ skewness matrix $\vecLambda$. The distribution is also parameterized by a $p$-dimensional location parameter $\vecmu$, $p\times p$ scale matrix $\vecSigma$, $\omega\in\mathbb{R}_+$, and $\lambda\in \mathbb{R}$.
 To derive this density, we make use of Proposition~1.

\begin{prop}  
Given $W\sim\text{GIG}(\psi,\chi,\lambda)$,  and $c\in\mathbb{R}^{q},$
\begin{equation} 
{\rm E}\left\{\Phi_q\left( \frac{c} {\sqrt{W}} \bigg| \vecDelta\right)\right\} = H_q \left\{ c \bigg(\frac{\psi}{\chi}\bigg)^{1/4} \bigg|\mathbf{0},\vecDelta, \lambda, \sqrt{\chi\psi},\sqrt{\chi\psi}  \right\}.
\end{equation}
\end{prop}
\noindent See Supplementary Material S.3 for a proof. 

%\subsubsection*{Proof:}
Then, to derive the HTH density, we write out the joint density of $\vecX$ and $W$ and then integrate out $W$. The conditional density is
$\vecX | w \sim\text{SN}_p( \vecmu,  w\vecSigma, \sqrt{w}\vecLambda)$ and $W\sim\text{GIG}(\omega,\omega,\lambda)$ which gives
\begin{equation*}
 \begin{split}
f&_{\text{HTH}}(\vecx|\vectheta)=\int_0^{\infty} f_{\tiny{\text{SN}}}(\vecy\mid\vecmu,  w \vecSigma, \sqrt{w}\vecLambda) {\color{black}f_{\text{GIG}}}(w | \omega, \omega, \lambda) d w \\
&=\int_0^{\infty} 2^q\phi_p(\vecy\mid\vecmu, w \vecOmega)\ \Phi_q\left\{ \frac{1}{\sqrt{w} } \vecLambda^{\top}\vecOmega^{-1}(\vecy-\vecmu)\Big|\vecDelta \right\} {\color{black}f_{\text{GIG}}}(w | \omega, \omega, \lambda) d w \\
&= 2^qh_p(\vecx|\vecmu,\vecOmega,\lambda,\omega,\omega) \int_0^{\infty} \Phi_q\left\{ \frac{1}{\sqrt{w} } \vecLambda^{\top}\vecOmega^{-1}(\vecy-\vecmu)\Big|\vecDelta \right\} {\color{black}f_{\text{GIG}}}\left(w | \omega + \delta(\vecx,\vecmu|\vecOmega) , \omega, \lambda \right) d w \\
&= 2^qh_p(\vecx|\vecmu,\vecOmega,\lambda,\omega,\omega) \mbox{E}\left[ \Phi_q\left\{ \frac{1}{\sqrt{W} } \vecLambda^{\top}\vecOmega^{-1}(\vecy-\vecmu)\Big|\vecDelta \right\} \right].
\end{split}
\end{equation*}
Then, applying Proposition 1, we obtain the HTH density given in {\color{black}\eqref{HTHdens}}.

\subsection{Hierarchical Representation}\label{sec:hrep}

A random variable $\vecX\sim\text{HTH}_p(\vecmu,\vecSigma,\vecLambda,\lambda, \omega,\omega)$ can be represented hierarchically by
\begin{equation*}
\begin{split}
\vecX~|~\vecu,w&\sim\mathcal{N}_p\left(\vecmu+\vecLambda\vecU,w\vecSigma\right),\quad
\vecU~|~w\sim\text{HN}_q\left(w\mathbf{I}_q\right),\quad
W\sim\text{GIG}(\omega,\omega,\lambda).
\end{split}
\end{equation*}
Herein, write $\vecr=\vecLambda^{\top}\vecOmega^{-1}(\vecx-\vecmu)$ and $\vecdelta(\vecx~|~\vecmu,\vecOmega)=\left(\vecx-\vecmu\right)^{\top}\vecOmega^{-1}\left(\vecx-\vecmu\right)$. There follows a listing of the joint and conditional densities pertaining to the variables $\vecX$, $\vecU$, and $W$:
\begin{equation*}
 \begin{split} 
& f_{\vecX,\vecU,W}(\vecx,\vecu,w)=\frac{w^{\lambda-(p+q)/2-1}}{\pi^{(p+q)/2}2^{(p-q)/2+1}K_\lambda(\omega)|\vecSigma|^{1/2}}\\
 &\qquad\qquad  \times\exp \left[-\frac{1}{2w}\left\{\omega w^2+\omega+(\vecu-\vecr)^{\top}\vecDelta^{-1}(\vecu-\vecr)+\vecdelta(\vecx|\vecmu,\vecOmega) \right\}\right],\\
&f_{\vecX,\vecU}(\vecx,\vecu)=\frac{|\vecSigma|^{-1/2}}{\pi^{(p+q)/2}2^{(p-q)/2}K_\lambda(\omega)}\left\{\frac{\omega+(\vecu-\vecr)^{\top}\vecDelta^{-1}(\vecu-\vecr)+\vecdelta(\vecx|\vecmu,\vecOmega)}{\omega}\right\}^{\{\lambda-(p+q)/2\}/2}\\
&\qquad\qquad\times K_{\lambda-(p+q)/2}\left(\sqrt{\omega\left\{\omega+(\vecu-\vecr)^{\top}\vecDelta^{-1}(\vecu-\vecr)+\vecdelta(\vecx|\vecmu,\vecOmega)\right\}}\right),\\
\end{split}\end{equation*}
\begin{equation*}
 \begin{split} &f_{\vecX,W}(\vecx,w) = \frac{w^{\lambda-p/2-1}|\vecDelta|^{1/2}}{(\pi)^{p/2}2^{p/2-q+1}K_\lambda(\omega)|\vecSigma|^{1/2}}\\
& \qquad\qquad\times\exp \left[-\frac{1}{2w}\left\{\omega w^2+\omega+\vecdelta(\vecx|\vecmu,\vecOmega)\right\}\right]\vecPhi_q\{(1/\sqrt{w})r|\vecDelta\},\\
&f_{\vecU|\vecx,w}(\vecu~|~\vecx,w) = \frac{w^{-q/2}}{(2\pi)^{q/2}|\vecDelta|^{1/2}\vecPhi_q\{(1/\sqrt{w})r|\vecDelta\}}\exp\left[-\frac{1}{2w}\left\{(\vecu-\vecr)^{\top}\vecDelta^{-1}(\vecu-\vecr)\right\}\right],\\
&f_{W|\vecx}(w~|~\vecx)=\frac{w^{\lambda-(p/2)-1}}{2K_{\lambda-(p/2)}\big\{\sqrt{\omega(\omega+\vecdelta(\vecx|\vecmu,\vecOmega))}\big\}}\bigg\{\frac{\omega}{\omega+\vecdelta(\vecx|\vecmu,\vecOmega)}\bigg\}^{\{\lambda-(p/2)\}/2}\\
&\qquad\qquad\times\exp\bigg[\frac{-1}{2w}\big\{\omega w^2 + \omega +\vecdelta(\vecx|\vecmu,\vecOmega)\big\}\bigg]\vecPhi(\vecr/\sqrt{w}\mid \vecDelta)\\
&\qquad\qquad\div H_q\left[\vecr\left\{\frac{\omega}{\omega+\vecdelta(\vecx|\vecmu,\vecOmega)}\right\}^{1/4}\bigg|\mathbf{0},\mathbf\vecDelta,\lambda-(p/2),\sqrt{\omega\{\omega+\vecdelta(\vecx|\vecmu,\vecOmega)\}}\right],
\end{split}
\end{equation*}
and $W~|~\vecx,\vecu\sim \text{GIG}\{\omega,\omega + (\vecu-\vecr)^{\top}\vecDelta^{-1}(\vecu-\vecr)+\vecdelta(\vecx|\vecmu,\vecOmega),\lambda-(p+q)/2\}$.
Note that the conditional density of $\vecU$ given $\vecX=\vecx$ is
\begin{equation}
\begin{split}
f_{\vecU|\vecx}(\vecu|\vecx)=\frac{1}{c_\lambda} h_q(\vecu|\vecr,\vecDelta,\lambda-p/2,\gamma,\gamma),
\end{split}
\label{ux}
\end{equation}
where the support of $\vecU$ is $\mathbb{R}_+^q$ %, i.e., the positive plane of $\mathbb{R}_q$ 
and $$c_\lambda= H_q\left[\vecr\left\{\frac{\omega}{\omega+\vecdelta(\vecx|\vecmu,\vecOmega)}\right\}^{1/4}\bigg|\mathbf{0},\mathbf\vecDelta,\lambda-(p/2),\gamma,\gamma \right].$$
It follows that {\color{black}$\vecU~|~w,\vecx \sim \text{TN}(\vecr,w\vecDelta; \mathbb{R}^{q}_+),$
where $\text{TN}(\cdot)$ denotes the truncated normal distribution}. Also, $\vecU~|~\vecx \sim \text{TH}_q(\vecr, \vecDelta,\lambda-p/2, \gamma,\gamma);\mathbb{R}_+^q).$  Here, $\text{TH}_q(\vecmu,\vecSigma, \lambda,\psi,\chi;\mathbb{R}_+^q)$ denotes the $q$-dimensional symmetric truncated hyperbolic distribution with density 
$$f_{\text{TH}}(\vecu~|~\vecmu,\vecSigma,\lambda,\psi,\chi;\mathbb{R}_+^q)=
\frac{h_q(\vecu~|~\vecmu,\vecSigma,\lambda,\psi,\chi)}{\int^{\infty}_{0}\ldots \int^{\infty}_{0}h_q(\vecu~|~\vecmu,\vecSigma,\lambda,\psi,\chi)d\vecu}\mathbb{I}_{\mathbb{R}_+^q}(\vecu),$$
where $\mathbb{I}_{\mathbb{R}_+^q}(\vecu)=1$ if
$\vecu\in \mathbb{R}_+^q$ and $\mathbb{I}_{\mathbb{R}_+^q}(\vecu)=0$ otherwise. In this way, the symmetric hyperbolic distribution is truncated to exist only within $\mathbb{R}_+^q$.

\subsection{Free Parameters and Identifiability}
A permutation matrix is a square matrix with exactly one entry of 1 in each row and each column and 0s elsewhere. Now, if $\mathbf{P}$ is a permutation matrix, then $\mathbf{P}^{\top}\vecU$ has the same distribution as $\vecU\sim\text{HN}\left(\ident_q\right)$. 
Therefore, $\vecLambda$ in our HTH distribution is unique up to permutations applied to the right, i.e., $f_{\text{HTH}}(\vecx \mid \vectheta) = f_{\text{HTH}}(\vecx \mid \vectheta^*)$ for all $\vecx\in\mathbb{R}^p$, where $\vectheta=(\vecmu,\vecSigma,\load,\lambda,\omega)$ and $\vectheta^*=(\vecmu,\vecSigma,\load^*,\lambda,\omega)$, if there exists a permutation matrix $\mathbf{P}$ such that $ \vecLambda^* =\vecLambda \mathbf{P} $.
Because $\vecLambda$ is unique up to permutations applied to the right, it has $p q$ free parameters --- the permutation matrix simply changes the position of the elements and not the actual values of the elements. Accordingly, sorting the $\vecLambda$ by the norm of the columns or some other sorting method  is needed to ensure that our factor-like model is identifiable, i.e., to take away the caveat about permutations applied to the right.

\subsection{Moments}

\subsubsection{Notation}

In this section, we derive the first two moments of the truncated symmetric hyperbolic distribution. These moments have not previously appeared in the literature. Here we introduce the notation used in this section. Consider an $n$-vector $\vecb$ and an $n\times n$ matrix $\matc$. Use $b_r$ to denote the $r$th element of $\vecb$, and let $\vecb_{-r}$ denote the $(n-1)$-vector that results from removing the $r$th element from $\vecb$. Let $\vecb_{rs}$ denote the 2-vector $(b_r,b_s)^{\top}$, and use $\vecb_{-rs}$ to denote the $(n-2)$-vector that results from removing the $r$th and $s$th elements from $\vecb$. Let $c_{rs}$ denote the element in the $r$th row and $s$th column of $\matc$, and use $\matc_{-rs}$ to denote the $(n-2)\times(n-2)$ matrix that results from removing the $r$th row and $s$th column from $\matc$. Use $\matc_{rs}$ to denote the $2\times 2$ matrix $$\begin{pmatrix}c_{rr}&c_{rs}\\c_{sr}&c_{ss}\end{pmatrix},$$ and let $[\matc]_{rs}$ denote the element in the $r$th row and $s$th column of $\matc$.

\subsubsection{Univariate Case}

Where $Y$ is a univariate random variable following the truncated symmetric hyperbolic distribution, i.e.,\ $ Y \sim \text{TH}_1(\mu,{ \color{black}\sigma^2},\lambda,\omega,\omega$), and truncated to exist on the domain $[l_1,l_2]$, the first and second moments of $Y$ are given by 
$$\textrm{E}(Y)=\mu+\sigma^2R_\lambda(\omega)\left\{\frac{h_1(l_1\mid \vectheta_{[\lambda+1]})-h_1(l_2\mid\vectheta_{[\lambda+1]})}{H_1(l_2\mid \vectheta_{[\lambda]})-H_1(l_1\mid \vectheta_{[\lambda]})}\right\}$$
and 
\begin{equation*}
\begin{split}
\textrm{E}(Y^2)=\sigma^2 R_\lambda(\omega)&\Bigg\{\frac{H_1(l_2\mid \vectheta_{[\lambda+1]})-H_1(l_1\mid \vectheta_{[\lambda+1]})}{H_1(l_2\mid  \vectheta_{[\lambda]})-H_1(l_1\mid  \vectheta_{[\lambda]})}+\frac{l_1h_1(l_1\mid \vectheta_{[\lambda+1]})-l_2h_1(l_2\mid \vectheta_{[\lambda+1]})}{H_1(l_2\mid \vectheta_{[\lambda]} )-H_1(l_1\mid \vectheta_{[\lambda]} )}\\
&\qquad\qquad\qquad\qquad\qquad\qquad\qquad-\frac{h_1(l_2\mid \vectheta_{[\lambda+1]})-h_1(l_1\mid\vectheta_{[\lambda+1]})}{H_1(l_2\mid  \vectheta_{[\lambda]})-H_1(l_1\mid  \vectheta_{[\lambda]})}\Bigg\},
\end{split}
\end{equation*}
respectively, where $\vectheta_{[\lambda]}=(\mu,\sigma^2,\lambda,\omega,\omega)$ and $\vectheta_{[\lambda+1]}=(\mu,\sigma^2,\lambda+1,\omega,\omega).$  Note that, to fit a HTH mixture model, we let $l_1=0$ and $l_2=\infty$.

\subsubsection{Multivariate Case}

Suppose $\vecY$ is a multivariate random variable following a truncated symmetric hyperbolic distribution, truncated to exist in the region $\mathbb{R}_{[a_1, \infty)\times\cdots\times[a_q, \infty)}^q$, i.e.,\ $\vecY \sim \text{TH}_q(\vecmu,\matsig,\lambda,\omega,\omega; \mathbb{R}_{[a_1, \infty)\times\cdots\times[a_q, \infty)}^q)$. The first moment of the multivariate truncated symmetric hyperbolic distribution is given by
$$
\textrm{E}(\vecY)=\vecmu+c^{-1}\vecSigma\vecepsilon,
$$
where \begin{equation}\label{eqn:c}c=\int^{\infty}_{a_1}\cdots \int_{a_q}^{\infty}h_q(\vecy \mid \vecmu,\vecSigma,\lambda,\omega,\omega)d\vecy,\end{equation} and the $r$th entry of $\vecepsilon$ is 
\begin{equation*}
\begin{split}
R_\lambda(\omega) &h_1(a_r \mid \mu_r,\sigma_{rr},\lambda+1,\omega,\omega)\\&\times
H_{q-1}\left(\veca_{-r}~\Bigg|~\vecmu_{-r}, \sqrt{\frac{\omega+\gamma(r)}{\omega}}\matsig_{-r}, \lambda+1/2, \omega\sqrt{1+\gamma(r)/\omega},\omega\sqrt{1+\gamma(r)/\omega}\right),
\end{split}
\end{equation*}
where $\veca=(a_1,\ldots,a_q)^{\top}$ and $\gamma(r)=(a_r-\mu_r)^2/\sigma_{rr}$.
The second moment is given by
$$\textrm{E}(\vecY \vecY^{\top})= \vecmu\textrm{E}(\vecY)^{\top} +\textrm{E}(\vecY)\vecmu^{\top} -\vecmu\vecmu^{\top} +\frac{k}{c} \, \vecSigma +\frac{1}{c}\vecSigma(\vecH+\vecD)\vecSigma,$$
where $c$ is as defined in \eqref{eqn:c}, $\vecH$ is a $q \times q$ matrix with zeros on the diagonal and $(r,s)${th} off-diagonal entry given by 
\begin{equation*}
\begin{split}
&\frac{K_{\lambda+2}(\omega)}{K_{\lambda}(\omega)} h_2(\veca_{rs} \mid \vecmu_{rs},\matsig_{rs},\lambda+2,\omega,\omega)\\
&\qquad\times H_{q-2}\left(\veca_{-rs}~\Bigg|~\vecmu_{-rs}, \sqrt{\frac{\omega+\gamma(r,s)}{\omega}}\matsig_{-rs}, \lambda+1, \omega\sqrt{1+\gamma(r,s)/\omega}, \omega\sqrt{1+\gamma(r,s)/\omega}\right),
\end{split}
\end{equation*}$\gamma(r,s)= (\veca_{rs}-\vecmu_{rs})^{\top}\vecSigma_{rs}^{-1}(\veca_{rs}-\vecmu_{rs})$, $\vecD$ is a $q \times q$ diagonal matrix with the $r${th} diagonal entry given by $\sigma_{rr}^{-1}\{(a_r-\mu_r)\epsilon_r-[\vecSigma \vecH]_{rr}\}$, and 
$$k= R_\lambda(\omega)\int_{a_1}^{\infty}\cdots\int_{a_q}^{\infty} h_q(\vecy \mid \vecmu, \vecSigma, \lambda+1, \omega,\omega)d\vecy.$$

For a mixture of HTH distributions, we set $\veca=\mathbf{0}.$ Note that in the bivariate case, the second moment is equivalent to the second moment in the multivariate case $(p>2)$ with the exception that $\vecH$ is a $q \times q$ matrix with zeros on the diagonal and the $(r,s)${th} off-diagonal entry given by 
$$\frac{K_{\lambda+2}(\omega)}{K_{\lambda}(\omega)} \, h_2(\veca_{rs} \mid \vecmu_{rs},\matsig_{rs},\lambda+2,\omega,\omega).$$
See Supplementary Material S.1 for proofs pertaining to the moments derived in this section. Note that results on double-truncation for the multivariate $t$-distribution \cite{ho12} could be extended to the multivariate truncated symmetric hyperbolic distribution.

\subsection{Special cases}\label{sec:special}

The HTH distribution can be considered as a special case of the canonical fundamental skew-spherical distribution \citep{arellano05}. 
%The HTH distribution contains several distributions as special or limiting cases including the skew-$t$, multivariate-$t$, Gaussian, and variance-Gamma distributions.  
Naturally, for all asymmetric distributions that exist as special or limiting cases of the HTH distribution, the value of $q$ may vary in the skewness parameter. Arellano-Valle and Genton \cite{arellano05} introduced a canonical fundamental skew-normal distribution (CFUSN) whereby skewness is modeled by a $p\times q$ skewness matrix. They also introduce a canonical skew-$t$ distribution (CFUST) that can capture both the classical and SDB versions of the skew-$t$ and skew-normal distributions as special cases. Our HTH distribution includes the CFUST distribution as a special case and, as a consequence, the CFUSN as a limiting case. %{\color{red}The HTH distribution also contains the  shifted asymmetric Laplace, inverse Gaussian, multivariate-$t$ and multivariate Gaussian distributions as special or limiting cases.} 

\section{A Finite Mixture of HTH Distributions}\label{sec:ident}

\subsection{A Finite Mixture Model}
To facilitate the modelling of heterogeneous data in general, we develop a finite mixture of HTH distributions. The model has density 
\begin{equation}
g(\vecx~|~\vecvarthet)=\sum_{g=1}^{G}\pi_g f_{\text{HTH}}(\vecx~|~\vectheta_g),
\label{eq:generalmixture1}
\end{equation}
such that 
%\begin{center}
$\vecvarthet=(\pi_1,...,\pi_G,\vectheta_1,...,\vectheta_G)$, 
%\end{center}
$\pi_g>0$ is the $g$th mixing proportion such that $\sum_{g=1}^{G}\pi_g=1$, $\vectheta_g=(\vecmu_g,\vecSigma_g,\vecLambda_g,\lambda_g,\omega_g)$, and $f_{\text{HTH}}(\vecx|\vectheta_g)$ is the density of the HTH distribution. 

\subsection{Finite Mixture Identifiability}\label{sec:ident}

We prove the identifiability of our model and, therefore, identifiability of all mixture models that exist as special cases of our model. 
The identifiability of mixtures of distributions has been studied by \cite{teicher63} and \cite{holzmann06}, among others. In the fitting of mixture models, identifiability is necessary for consistent estimation of the parameters of the mixing distribution. \cite{holzmann06} point out that previous literature was lacking formal proofs of identifiability of mixture models, and they provide proofs of identifiability for elliptical finite mixtures. 

{\color{black}\begin{defn}\label{def:1}
A mixture distribution has the form
%The transformation
\begin{equation}\label{eq:mixturedist}
H( \vecx | \vectheta ) = \int_{\mathcal{S}_y } F(\vecx~|~\vecy, \vectheta_F  ) dG(\vecy | \vectheta_G ),
\end{equation}
where $F( \vecx| \vecy, \vectheta_F )$ is a distribution function for all $\vecy \in \mathcal{S}_y$, $\vectheta = ( \vectheta_F, \vectheta_G)  \in \Omega$  and $G$ is a distribution function defined on $ \mathcal{S}_y$, and the mixture density $H$ is said to be identifiable if and only if there is a unique $G$ yielding $H$.\end{defn}}

One famous example is the $t$-distribution being a scale mixture of a Gaussian distribution where $F$ is Gaussian with mean $\mu$ and variance $\sigma^2/y $, $G$ is a gamma distribution with parameters $(\nu/2, \nu/2)$ where $\nu$ is the degrees of freedom. That is $\vectheta_F = (\mu, {\color{black}\sigma^2})$ and $\vectheta_G = (\nu)$. Similar to this example the $H$ mixture distribution \eqref{eq:mixturedist} can depend on parameters through $F$ and $G$ and these parameters are identifiable if there is a unique $G$ yielding $H$. 

Other examples include the skew-normal density \eqref{eq:sndensity}, its skew-$t$ analogue, and the HTH density \eqref{HTHdens}, all hidden truncated densities that can be represented as a mixture distribution
\begin{equation*}
h(\vecx \mid \vectheta)
=\frac{1}{k_q} \int_{\vecy \in \mathbb{R}^q_+ } h(\vecx, \vecy \mid \vectheta) \mathrm{d} \vecy= \frac{1}{k_q} \int_{\vecy \in \mathbb{R}^q_+ } h(\vecx |\vecy, \vectheta)g( \vecy) \mathrm{d} \vecy,
\end{equation*}
where
\begin{equation*}
 k_q =\int_{\vecx \in \mathbb{R}^p } \int_{\vecy \in \mathbb{R}^q_+ } h(\vecx, \vecy \mid \vectheta) \mathrm{d} \vecyÊÊ\mathrm{d}\vecx
%= \int_{\vecy \in \mathbb{R}^q_+ } h( \vecy ) \mathrm{d} \vecyÊ 
= \frac{1}{2^q},
\end{equation*}
and $\vecy$ is a realization of a $q$-dimensional random variable $\vecY$. Note that the integrating constant $k_q$ does not depend on the parameter $\vectheta$ and the support of the latent variable $\vecY$ is the hypercube defined by the positive plane $\mathbb{R}^q_+$. Also, for these models $g(\vecy)$ does not depend on any parameters. 

Another example but special case of {\color{black}Definition~\ref{def:1}} is when the distribution function $G(\cdot)$ consists of a finite number of elements, that is, $\mathcal{S}_y = \left\{ \vectheta_1, \vectheta_2, \ldots, \vectheta_G\right\}$. Then
$$ H(\vecx~ |~\vecvarthet) =  \int_{ \Omega } F(\vecx | \vecy ) dG(\vecy)  = \sum_{g=1}^G  F(\vecx ~|~ \vecy) P(\vecy= \vectheta_g) =\sum_{g=1}^G \pi_g  F(\vecx ~|~ \vectheta_g),  $$
which is a finite mixture where $\pi_g  >0$, $\sum_{g=1}^G \pi_g =1$, and $\vecvarthet =  \left(\pi_1, \ldots, \pi_G, \vectheta_1, \ldots, \vectheta_G\right)$. If, in this special case, $H$ is found to be identifiable then $H$ is said to be finite mixture identifiable or equivalently is said that a finite mixture of $F(\vecx ~|~ \vectheta_g)$ is identifiable. 
\begin{thm}
A finite mixture of  $H(\vecx|\vectheta)$ is identifiable if the distribution  
$$H(\vecx |\vectheta)   =  \int_{ \mathcal{S}_y} F(\vecx | \vecy,\vectheta) dG(\vecy),$$ 
is identifiable where $G(\cdot)$ is a distribution function defined on $\mathcal{S}_y$ and a finite mixtures of $F(\vecx | \vecy, \vectheta)$ is identifiable for all $y \in \mathcal{S}_y$ and any $x \in \mathcal{S}_x$.\end{thm}%
\begin{proof} A finite mixture of $H(\vecx|\vectheta)$ is given by
%denoted by $K(\vecx|\vecvarthet)$ where $\vecvarthet=(\vectheta_1,\ldots,\vectheta_G,\pi_1,\ldots,\pi_G)$ is 
% $$ H(\vecx ~|~ \vecvarthet )=  \sum_{g=1}^n \pi_g H(\vecx|\vectheta_g) =  \sum_{g=1}^n \pi_g   \int_{ \mathcal{S}_y} F(\vecx|\vecy,\theta_g) dG(\vecy) = \int_{ \mathcal{S}_z } F(x | \vecu ) dG^*(\vecu),$$
\begin{equation*} \begin{split}
%K(\vecx ~|~ \vecvarthet )&=
\sum_{g=1}^G \pi_g  H(\vecx|\vectheta_g)
&=  \sum_{g=1}^G \pi_g   \int_{ \mathcal{S}_y} F(\vecx|\vecy,\vectheta_g) dG(\vecy)  \\
& =  \sum_{g=1}^G \int_{ \mathcal{S}_y} F(\vecx|\vecy, \vecz) dG(\vecy) P(\vecZ=\vectheta_g) \\
% & =  \int_{ \mathcal{S}_y}  \sum_{g=1}^n F(\vecx|\vecy, z) P(\vecz=\theta_g) dG(\vecy)  \\
& =  \int_{ \mathcal{S}_y} \sum_{g=1}^G  F(\vecx|\vecy, \vecz)  P(\vecZ=\vectheta_g) dG(\vecy) 
 = \int_{ \mathcal{S}_u } F(\vecx | \vecu ) dG^*(\vecu),
\end{split}  \end{equation*}
where $G^*(\vecu)  $ is the distribution function for $\vecu = ( \vecy, \vecz)$  defined over $\mathcal{S}_u =\mathcal{S}_y \times \vecTheta$ and $\vecTheta = \left\{ \vectheta_1, \ldots, \vectheta_G \right\}$. $G(\vecy)$ is unique because $H(\vecx|\vectheta)$ is identifiable. The measure on $\vecz$ is unique because a finite mixture of $F(\vecx | \vecy, \vectheta)$ is identifiable. Therefore, $G^*(\vecu)$ is unique and so a finite mixture of $H(\vecx|\vectheta)$ is identifiable.
\end{proof}
\begin{cor}
Finite mixtures of the skew-normal, skew-t, and HTH distributions are identifiable. %In addition. all hidden truncated densities derived from elliptical densities in which finite mixtures of them are identifiability.}\\
\end{cor}
\noindent The proof of Corollary 1 follows by Theorem~1 and the following facts. Skew-elliptical or hidden truncated elliptical distributions are identifiable mixture distributions \citep[see][]{pyne09, sahu03}. Multivariate linear regression using finite mixtures of Gaussian and $t$ distributions are identifiable \citep[see][]{galimberti11,galimberti14}. Note that multivariate linear regression using finite mixtures of symmetric hyperbolic distributions and/or elliptical distributions can be shown to be identifiable using a similar technique to \cite{galimberti14} and thus is omitted for brevity. 

A section on identifiability would not be complete without some discussion of the well-known label switching problem.  The term label switching is used by \cite{render84} in reference to ``the invariance of the likelihood under relabelling of the mixture components" \citep{stephens00}. As \cite{stephens00} points out, label switching can lead to difficulties when model-based clustering is carried out within the Bayesian paradigm. \cite{yao09}, \cite{celeux00}, and \cite{yao12} also discuss label switching in the Bayesian mixture context. In the present work, the maximum likelihood inferential paradigm is used and label switching has no practical implications and arises only as a theoretical identifiability issue that can usually be resolved by specifying some ordering on the mixing proportions, e.g., $\pi_1>\pi_2>\cdots>\pi_G$. Note that in cases where mixing proportions are equal, a total ordering on other model parameters can be considered.

\section{s-concavity of the Truncated Hyperbolic Distribution}\label{sec:concave}

When using mixtures of non-elliptical distributions for clustering, it is important to understand whether the component contours are convex. If not, it will be possible that one component corresponds to multiple clusters, e.g., an x-shaped component, and this leads to results that are difficult to interpret. Note that this is especially problematic when the data are not sufficiently low-dimensional to visually understand what is happening. In this section, the convexity, or quasi-concavity, of the truncated hyperbolic (TH) distribution is discussed.

%We want densities that are quasi-concave. 
A function $f(\vecx)$ is quasi-concave if each upper-level set $$S_\alpha(f) = \{\vecx|f(\vecx) \geq \alpha\}$$ is a convex set. If the density is quasi-concave then it is unimodal. 
\cite{tortora15} point out that the generalized hyperbolic distribution has a quasi-concave density and, if $\lambda > (p+1)/2$, it has a log-concave density. Because the HTH distribution is obtained by integrating out a set of variables from the symmetric hyperbolic distribution, we can show that the HTH distribution is quasi-concave if the symmetric hyperbolic density is $s$-concave or log-concave --- note that, in general, $s\in\mathbb{R}$.  
Next, we will show that the symmetric generalized hyperbolic distribution is an s-concave density. In particular, we will show that the density is $-1/p$-concave when $\lambda \le - p/2 $.
A density, $f$, is s-concave on an open convex set $\mathcal{C} \subset \mathbb{R}^m$ if, for every $\vecx_0$ and $\vecx_1$ in $\mathcal{C}$ and $\alpha \in (0,1)$, we have
\begin{equation*}
f\left\{\alpha \vecx_0 + (1-\alpha)\vecx_1\right\}
\ge \left\{\alpha f\left(\vecx_0 \right)^s  + (1-\alpha)f\left(\vecx_1\right)^s\right\}^{1/s}.
\end{equation*}
For a discussion of s-concave densities, see \cite{dharmadhikari88}.

We begin with the fact that the symmetric hyperbolic distribution is proportional to the composition $f\{g(\vecx)\} = g(\vecx)^\tau K_\tau\{g(\vecx)\}$, where $\tau = \lambda-p/2$ and the function $g(\vecx) = \sqrt{a + b \times \delta\left(\vecx, \vecmu| \vecSigma\right)}$ is convex with $a$ and $b$ both positive. Then, we apply the following theorem. 
\begin{thm}
If $f$ is a non-negative and convex function defined on the convex set $\mathcal{C}$ and h is $s$-concave and monotone decreasing then $h\{f(\vecx)\}$ is a $s$-concave function. 
\end{thm}
\begin{proof}
Let $\vecx_0$ and $\vecx_1 \in \mathcal{C}$ and $\alpha \in (0,1)$, then $f\left\{\alpha \vecx_0 + (1-\alpha)\vecx_1\right\} 
\le  \alpha f\left(  \vecx_0 \right)  + (1-\alpha) f\left(\vecx_1\right) $. This implies 
\begin{eqnarray*}
f\left\{\alpha \vecx_0 + (1-\alpha)\vecx_1\right\} 
&\le&  \alpha f\left(  \vecx_0 \right)  + (1-\alpha) f\left(\vecx_1\right) \\ 
h\{f\left( \alpha \vecx_0 + (1-\alpha)\vecx_1\right)\}
&\ge&  h\left\{\alpha f\left(  \vecx_0 \right)  + (1-\alpha) f\left(\vecx_1\right)\right\}\\
h\{f\left( \alpha \vecx_0 + (1-\alpha)\vecx_1\right)\} &\ge&  \left[\alpha h\{f\left(  \vecx_0 \right)\}^s  + (1-\alpha)h\{f\left(\vecx_1\right)\}^s  \right]^{1/s}. 
\end{eqnarray*}
\end{proof}
A density $f$ is s-concave for $s \in (0,\infty)$ if and only if $f^s$ is convex. 
Then we are required to show that the function, $f,$ is s-concave. The function $f(u) = u^{\tau} K_{\tau}(u)$ is s-concave for $s \in (0,\infty)$ if and only if the function $q(u)= f(u)^s$ is convex where $s <0$.
%\begin{equation*}   f(u) = u^{\tau} K_{\tau}(u) \end{equation*}
%\begin{equation*}   f'(u) = \tau u^{\tau-1} K_{\tau}(u) +  u^{\tau} K'_{\tau}(u) \end{equation*}
%\begin{equation*}  f(u)^s = u^{\tau s} K_{\tau}(u)^s \end{equation*}
%\begin{equation*} q'(u)  = \tau s u^{\tau s-1} K_{\tau}(u)^s + s u^{\tau s}  K_{\tau}(u)^{s-1} K'_{\tau}(u)\end{equation*}
The first derivative of $q$ can be written as 
\begin{equation*} 
q'(u)  = s f(u)^{s-1} f'(u) 
= s u^{\tau s } K_{\tau}(u)^{s}
\left\{\frac{\tau}{u} +  \frac{ K'_{\tau}(u)  }{ K_{\tau}(u)}\right\} 
 = s q(u)
\left\{\frac{\tau}{u} +  \frac{ K'_{\tau}(u)  }{ K_{\tau}(u)}\right\} 
\end{equation*} 
and second derivative of $q$ has the relation 
\begin{equation*} 
 \frac{u^2}{s q(u) } q''(u) =  (s+1) \tau^2  +u^2 -  \tau+ (2 \tau s -1)   u \frac{ K'_{\tau}(u)  }{ K_{\tau}(u) }  +  (s-1)  \left\{\frac{ u  K'_{\tau}(u) }{ K_{\tau}(u)}\right\}^2. 
\end{equation*} 
From \cite{baricz09}, we have  the inequality $u K_\tau'(u)/K_\tau(u)  < -\sqrt{u^2 + \tau^2}$, which holds for all $u >0$ and $\tau \in \mathbb{R}$; applying this inequality gives the following bound on the second derivative when $\tau s - 1 \le 0$: 
 %and from \cite{baricz2015} we have $u K_\tau'(u)/K_\tau(u)  < -\tau$  
\begin{equation*} 
 \frac{u^2}{s q(u) } q''(u) <   \tau( \tau s  - 1) + ( \tau s -1)  \left\{u\frac{ K'_{\tau}(u)}{K_{\tau}(u)}\right\}  \le 0.
\end{equation*}
%\begin{eqnarray*} 
% \frac{u^2}{s q(u) } q''(u) & < &  (s+1) \tau^2  +u^2 -  \tau+ (2 \tau s -1)   u \frac{ K'_{\tau}(u)  }{ K_{\tau}(u) }  +  s  \left[  \frac{ u  K'_{\tau}(u) }{ K_{\tau}(u) } \right]^2  - u^2 - \tau^2 \\ 
% & < &  s \tau^2   -  \tau+ (2 \tau s -1)   u \frac{ K'_{\tau}(u)  }{ K_{\tau}(u) }  +  s  \left[  \frac{ u  K'_{\tau}(u) }{ K_{\tau}(u) } \right]^2   \\
% & < &  s \tau^2   -  \tau+ (2 \tau s -1)  \left[ u \frac{ K'_{\tau}(u)  }{ K_{\tau}(u) } \right]  - \tau  s  \left[  \frac{ u  K'_{\tau}(u) }{ K_{\tau}(u) } \right]   \\ 
% & < &  \tau( \tau s  - 1) + ( \tau s -1)  \left[ u \frac{ K'_{\tau}(u)  }{ K_{\tau}(u) } \right]  \le 0
%\end{eqnarray*}
%\begin{equation*} 
% \tau s  - 1 \ge  0 
% \quad \quad \Rightarrow
% \tau s  \ge   1
% \quad \quad \Rightarrow
%s  \le  \frac{1}{\tau} .
%\end{equation*}
This implies that $s \le  1/\tau $, which means that we need $\lambda \le -p/2$ for the symmetric hyperbolic density to be $-1/p$-concave. 

{\color{black}The HTH is formed by integrating out $q$ latent variables from a symmetric hyperbolic distribution with $p + q$ variables, as shown in Section~\ref{sec:hrep}.} We have two cases:
\begin{itemize}
\item If $\lambda < -(p+q)/2$, the symmetric hyperbolic density is $-1/(p+q)$-concave. Then, by Theorem~3.21 in \cite{dharmadhikari88}, we have that the truncated hyperbolic distribution is $s^*$-concave, where
$$ s^* = \frac{s}{1+ps} = \frac{s}{1+ps}  =  \frac{-1/(p+q) }{1-q/(p+q)} =   \frac{-1 }{(p+q)-q} = - \frac{1}{p}.$$ 
\item If $\lambda > (p+q+1)/2$, the symmetric hyperbolic density is log-concave which implies the truncated hyperbolic distribution is also log-concave. 
\end{itemize}
For $-(p+q)/2 < \lambda < (p+q+1)/2$, concavity cannot be proven. However, as of yet we have not been able to identify any parameter set that leads to non-concave density contours. Refer to the appendix for examples of contours generated for $\lambda$ in this range. 

\section{Parameter Estimation}\label{sec:param}

An expectation-conditional maximization (ECM) algorithm \citep{meng93} is developed for parameter estimation for the mixture of HTH distributions. All maximum likelihood estimates for the model parameters are derived from the complete-data log-likelihood 
\begin{equation*}
\begin{split}
l_c&(\vecvarthet)%|\vecX,\vecU,W,Z)
=  \sum^{G}_{g=1}\sum^{n}_{i=1}z_{ig} \bigg[\ln\pi_g-\frac{1}{2}\ln|\vecSigma_g|-\ln K_{\lambda_g}(\omega_g)+\{\lambda_g-(p+q)/2-1\}\ln w_{ig}\\& \ \ -\frac{1}{2w_{ig}}\Big\{\omega_g w_{ig}^2+\omega_g+(\vecx_i-\vecmu_g-\vecLambda_g\vecu_{ig})^{\top}\vecSigma_g^{-1}(\vecx_i-\vecmu_g\vecLambda_g\vecu_{ig})+\vecu_{ig}^{\top}\vecu_{ig}\Big\}\bigg]+C,
\end{split} 
\end{equation*}
where $C$ is a constant with respect to the model parameters and $z_{ig}$ denotes component membership so that $z_{ig}=1$ if observation $i$ belongs to component $g$ and $z_{ig}=0$ otherwise. The algorithm alternates between the following two steps.  

\subsubsection*{E-step} 

Let $R_\lambda(\omega)=K_{\lambda+1}(\omega)/K_\lambda(\omega)$. To obtain the maximum likelihood estimates of the model parameters, we require the expectations $\textrm{E}\left(W_{ig}|\vecx_i, z_{ig}=1\right)\equalscolon a_{ig}$, $\textrm{E}\left({1}/{W_{ig}}|\vecx_i,z_{ig}=1\right)\equalscolon b_{ig}$, and $\textrm{E}(\ln(W_{ig})|\vecx_i,z_{ig}=1)\equalscolon c_{ig}$. A method for estimating  $\textrm{E}\{\ln W_{ig}|\vecx_i,z_{ig}=1\}$ via series expansions is detailed in Supplementary Material S.5, and additional information pertaining to the E-step calculations is given in Supplementary Material~S.4.
In addition, we need $\textrm{E}\{(1/W_{ig})\vecU_{ig}\mid\vecx_i,z_{ig}=1\}\equalscolon\vecd_{ig}$ and $\textrm{E}\{(1/W_{ig})\vecU_{ig}\vecU_{ig}^{\top}\mid\vecx_i,z_{ig}=1\}\equalscolon\vecE_{ig}.$ Note that it is convenient to use the relationships $$\textrm{E}\left\{(1/W_{ig})\vecU_{ig} \mid \vecx_i,z_{ig}=1\right\}=\textrm{E}\left\{(1/W_{ig}) \mid \vecx_i,z_{ig}=1\right\}\textrm{E}\left(\vecS_{ig} \mid \vecx_i,z_{ig}=1\right)$$ and $$\textrm{E}\left\{(1/W_{ig})\vecU_{ig}\vecU_{ig}^{\top} \mid \vecx_i,z_{ig}=1\right\}=\textrm{E}\left\{(1/W_{ig}) \mid \vecx_i,z_{ig}=1\right\}\textrm{E}\left(\vecS_{ig}\vecS_{ig}^{\top} \mid \vecx_i,z_{ig}=1\right),$$ where 
\begin{equation}
\begin{split}
\vecS_{ig}~|~\vecx_i \sim& \text{TH}_q\bigg(\vecLambda_g^{\top}\vecOmega_g^{-1}(\vecx_i-\vecmu_g), \sqrt{1+\vecdelta(\vecx_i|\vecmu_g,\vecOmega_g)/\omega}\vecDelta_g, \lambda_g-p/2-1,\\
& \qquad\qquad\omega_g\sqrt{1+\vecdelta(\vecx_i|\vecmu_g,\vecOmega_g)/\omega_g}, \omega_g\sqrt{1+\vecdelta(\vecx_i|\vecmu_g,\vecOmega_g)/\omega_g}~\bigg|~\mathbb{R}_+^q\bigg).
\end{split}
\label{eq:UX}
\end{equation}
See Supplementary Material S.2 for the proof of this result. 
At each E-step, we update the values of $a_{ig}$, $b_{ig}$, $c_{ig}$, $\vecd_{ig}$, and $\vecE_{ig}$. The component membership indicator variable $Z_{ig}$ is updated by
$$\textrm{E}(Z_{ig}|\vecx_i) = \frac{\pi_g f_{\text{HTH}}(\vecx_i|\vecmu_g,\vecSigma_g,\vecLambda_g,\lambda_g,\omega_g,\omega_g)}{\sum^{G}_{h=1}\pi_hf_{\text{HTH}}(\vecx_i|\vecmu_h,\vecSigma_h,\vecLambda_h,\lambda_h,\omega_h,\omega_h)}\equalscolon\hat{z}_{ig}.$$
As usual, all expectations are conditional on the current parameter estimates.

\subsubsection*{M-step}

At each M-step, we update the model parameters sequentially and conditionally on each other. The $g${th} location parameter is updated by
\begin{equation}
\vecmu_{g}=\frac{\sum^{n}_{i=1}\hat{z}_{ig}b_{ig}\vecx_i-\vecLambda_g\sum^{n}_{i=1}\hat{z}_{ig}\vecd_{ig}}{\sum^{n}_{i=1}\hat{z}_{ig}b_{ig}},
\end{equation}
and we update the $g${th} skewness matrix by $$\vecLambda_g=\vecM_{2g}^{\top}\vecM_{1g}^{-1},$$
where $\vecM_{1g}=\sum^{n}_{i=1}\hat{z}_{ig}\vecE_{ig}$, and $\vecM_{2g}=\sum^n_{i=1}\hat{z}_{ig}\vecd_{ig}(\vecx_i-\vecmu_g)^{\top}$. 
%\textbf{Ryan you wanted me to remind you to add something about the updates for $\omega$ and $\lambda$ here.}
The updates for $\omega_g$ and $\lambda_g$ are
$$\omega^{\new}_g= \omega_g -  \frac{  \partial_{\omega}  t }{ \partial^2_{\omega}  t }\bigg|_{\omega=\omega_g}\qquad
\text{and} \qquad\lambda_g^{\new}=  \overline{c}_g \lambda_g  \left\{ \frac{ \partial }{ \partial \lambda }  \ln K_{\lambda}\left(\omega_g\right) \bigg|_{\lambda=\lambda_g} \right\}^{-1},$$
respectively, where 
$$t_g(\omega_g,\lambda_g) = -\ln K_{\lambda_g} \left(\omega_g \right) + (\lambda_g-1)\overline{c}_g - \frac{\omega_g}{2} \left( \overline{a}_g + \overline{b}_g \right),$$
$\overline{a}_g=\sum^{n}_{i=1}\hat{z}_{ig}a_{ig}/n_g$, $\overline{b}_g=\sum^{n}_{i=1}\hat{z}_{ig}b_{ig}/n_g$, and $\overline{c}_g=\sum^{n}_{i=1}\hat{z}_{ig}c_{ig}/n_g$. Finally, the covariance parameter $\vecSigma_g$ is updated by 
\begin{equation}\label{eqn:sig}
\vecSigma_g=\frac{1}{n_g}\left\{\sum^{n}_{i=1}\hat{z}_{ig}b_{ig}(\vecx_i-\vecmu_g)(\vecx_i-\vecmu_g)^{\top} +\vecLambda_g \vecM_1\vecLambda_g^{\top} -\vecM_2^{\top}\vecLambda_g^{\top}-\vecLambda_g \vecM_2 \right\}.\end{equation}
Note that, by Jensen's inequality, $$\vecM_{1g}\succeq \sum^{n}_{i=1}z_{ig}b_{ig}\textrm{E}(\vecS_{ig}~|~\vecx_{ig},z_{ig}=1)\textrm{E}(\vecS_{ig}~|~\vecx_{ig},z_{ig}=1)^{\top},$$
{\color{black}where $\mathbf{A} \succeq \mathbf{B}$ means that $\mathbf{A} - \mathbf{B}$ is positive semi-definite.}  
Now if we replace $\vecM_{1g}$ by $ \sum^{n}_{i=1}z_{ig}b_{ig}\textrm{E}(\vecS_{ig}~|~\vecx_{ig},z_{ig}=1)\textrm{E}(\vecS_{ig}~|~\vecx_{ig},z_{ig}=1)^{\top}$ in \eqref{eqn:sig}, we get 
\begin{equation*}\begin{split}
\vecSigma_g^*=\frac{1}{n_g}&\sum^{n}_{i=1}\hat{z}_{ig}b_{ig}\left\{\vecx_i-\vecmu_g-\frac{1}{b_{ig}}\vecLambda_g\textrm{E}(\vecS_{ig}~|~\vecx_{ig},z_{ig}=1)\right\}\left\{\vecx_i-\vecmu_g-\frac{1}{b_{ig}}\vecLambda_g\textrm{E}(\vecS_{ig}~|~\vecx_{ig},z_{ig}=1])\right\}^{\top}\end{split}\end{equation*}
and it follows that $\vecSigma_g \succeq \vecSigma_g^* \succeq 0$. Therefore, $\vecSigma_g$ is positive semi-definite.

\section{Illustrations}\label{sec:results}

\subsection{Unboundedness of the Likelihood}

While a class of mixture densities may have an unbounded likelihood, a sequence of roots of the likelihood equation with the properties of consistency, efficiency, and asymptotic normality may still exist \citep[see][]{mclachlan00c}. When a class of mixtures is identifiable, regularity conditions can be given that ensure that this is the case \citep[see][]{render84}. As pointed out by \cite[][Section~2.5]{mclachlan00c}, these conditions are essentially multivariate analogues of the conditions given by \cite{cramer46} and so they should hold for many parametric families provided they are identifiable. Whether the HTH mixtures are identifiable, however, was an open question until we proved its identifiability here. To avoid degenerate solutions arising due to the unboundedness of the likelihood, we choose the root associated with the largest local maximum in our analyses. Note that other approaches exist to avoid spurious solutions. \cite{hathaway85} constrains the ratio of the smallest and largest variance parameters among the components in the case of a univariate normal distribution. \cite{chen08} use a penalized maximum likelihood estimator, and \cite{yao10} uses a profile likelihood approach.

\subsection{Starting Values and Model Selection}

For the data analyses, the group memberships are initialized using $k$-means clustering results. %5 sets of starting values were obtained by performing $k$-means clustering 5 times and the model parameters were initialized as described below. For the 5 sets of starting values, the log-likelihood was calculated and the set of values that obtained the largest log-likelihood value was used to initialize the ECM algorithm.
We initialize $\vecmu_g$ and $\vecSigma_g$ by the weighted mean and covariance matrix, respectively. We initialize $\vecLambda_g$ with values randomly generated from a normal distribution with mean 0 and standard deviation~1, and $\omega_g$ and $\lambda_g$ are initialized as 1. 

%\subsection{Model Selection}

The Bayesian information criterion \citep[BIC;][]{schwarz78} is used to select the number of components $G$ and the value of $q$. The BIC is given by 
%\begin{center}
$\text{BIC}=2l(\vecx,\hat{\vecvarthet})- \rho\ln n$,
%\end{center} 
where $l(\vecx,\hat{\vecvarthet})$ is the maximized log-likelihood, $\hat{\vecvarthet}$ is the maximum likelihood estimate of the model parameters $\vecvarthet$, $\rho$ is the number of free parameters in the model, and $n$ is the number of observations.  The BIC is commonly used for model selection in model-based clustering  with support given by \citep{campbell97,fraley02,leroux92}. 

\subsection{Simulation Studies}

The purpose of the first simulation study is to demonstrate the difference in fit between the HTHu and HTHm distributions. For this illustration, we initialized the models as described above and used a deterministic annealing algorithm \citep{ueda98,zhou09} to help overcome the issue of selecting good starting values. The deterministic annealing method flattens the likelihood surface, returning it to its original shape over several iterations of the ECM algorithm. This helps to avoid convergence to a local maxima rather than the global maximum. To implement deterministic annealing, the update for the group membership labels takes the form 
$$\textrm{E}(Z_{ig}~|~\vecx_i) = \frac{\{\pi_g f(\vecx_i~|~\vectheta_g)\}^d}{\sum^{G}_{h=1}\{\pi_hf(\vecx_i~|~\vectheta_h)\}^d},$$
where $d$ is increased at each iteration following a sequence of values from $d=0$ to $d=1$. 

First, data are simulated from a $G=1$ component mixture model with $n=250$, $p=2$, and $q=1$. 
%The parameters of the distribution were as follows:
%\begin{alignat*}{2}
%\vecmu & = (1,1)& \qquad \vecSigma_1 & = \begin{pmatrix} 1.5&-0.5\\-0.5&1.5\end{pmatrix} \\
%\lambda & = 3 & \qquad \omega&= 1.5\\
%\end{alignat*}
%and the entries of $\vecLambda$ were simulated from a normal distribution with mean 2 and standard deviation 2.
A HTH distributions is then fitted to the simulated data for $q=1$ (HTHu) and $q=2$ (HTHm), respectively. From the contour plots (Figure~\ref{fig:contour}), we see that both the HTHu and HTHm models obtain a very good fit to the data. Although, the contours of the HTHu model are softer compared to those of the HTHm model, both models capture the skewness well.
\begin{figure}[!h]
\centering
  \includegraphics[width=0.48\linewidth]{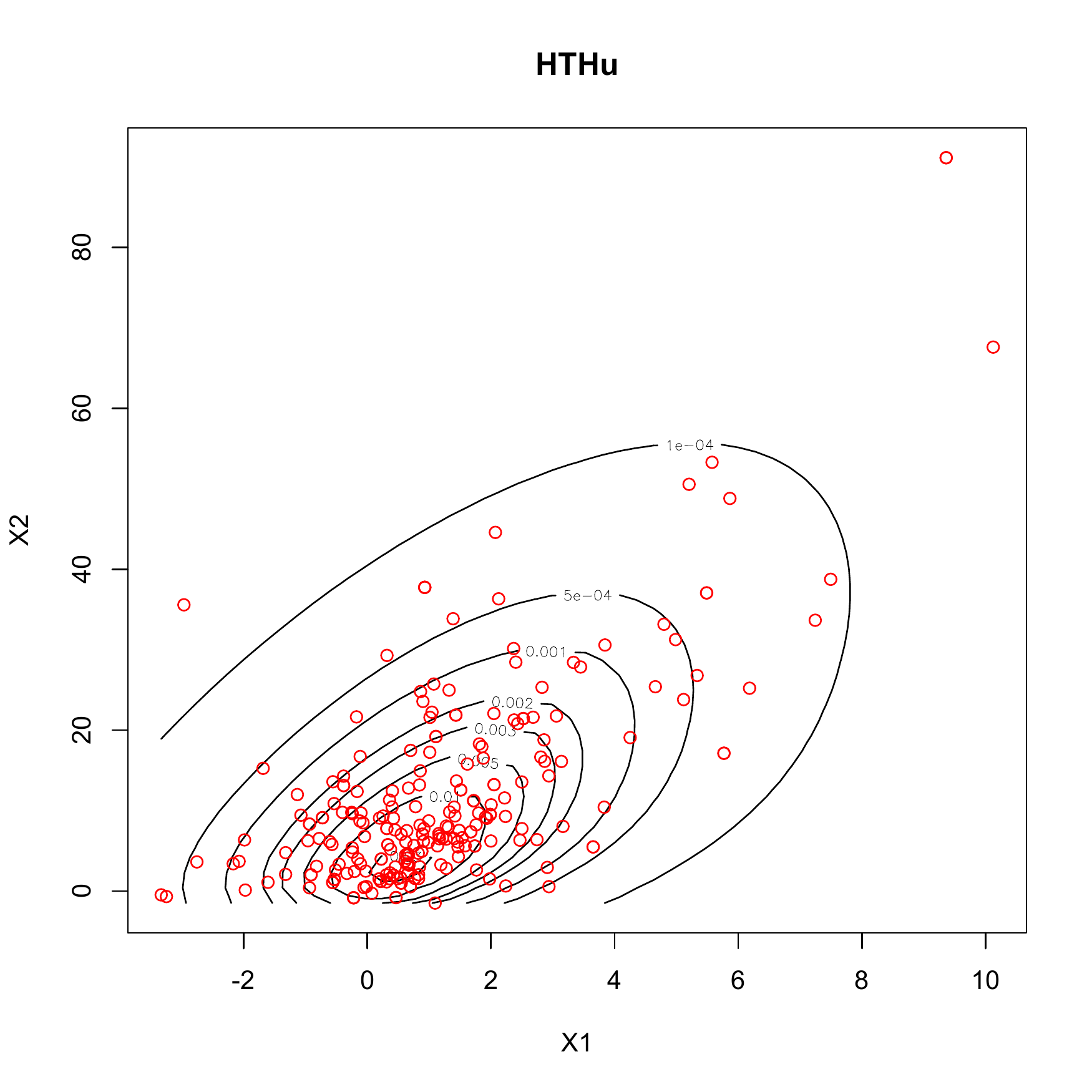} \quad
  \includegraphics[width=0.48\linewidth]{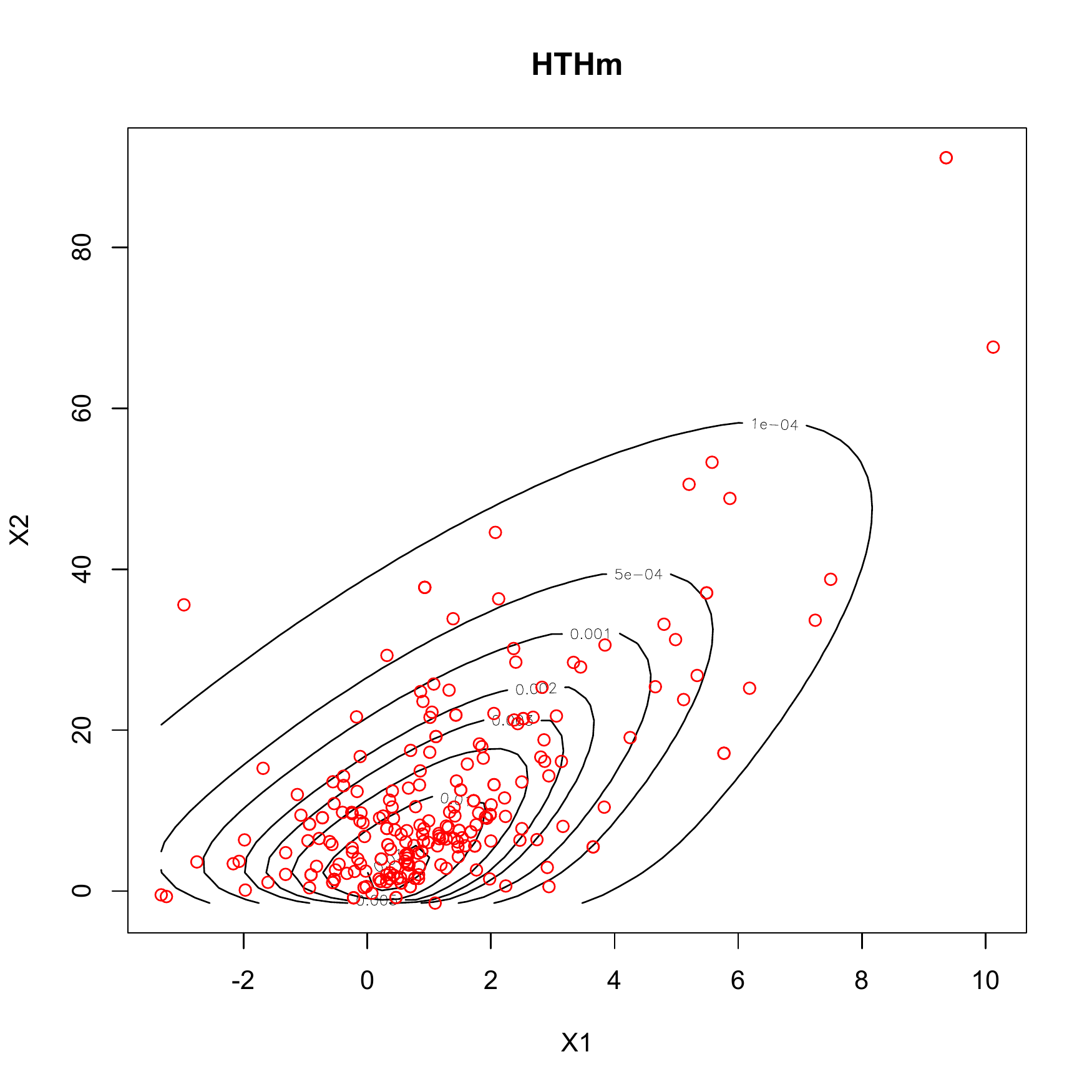}
  \caption{Contour plots for the HTHu and HTHm mixture models on the simulated data with $q=1$.}
  \label{fig:contour}
\end{figure}

Next, data is simulated from a $G=1$ component mixture model with $n=250$, $p=2$, and $q=2$. Again HTHu and HTHm distributions are fitted to these data. Looking at the resulting contour plots (Figure~\ref{fig:contour2}), we see a more drastic difference between the HTHu and HTHm models in the shapes of the contours. The HTHu model appears to have difficulty capturing the skewness in this instance and the contours take on a rounder shape. Given that the HTHm model has additional parameters in the skewness matrix, it may be expected that this model would be better able to capture the skewness in data simulated from a HTHu distribution than the HTHu model when fit to data simulated from a HTHm distribution. 
\begin{figure}[!ht]
\centering
  \includegraphics[width=0.48\linewidth]{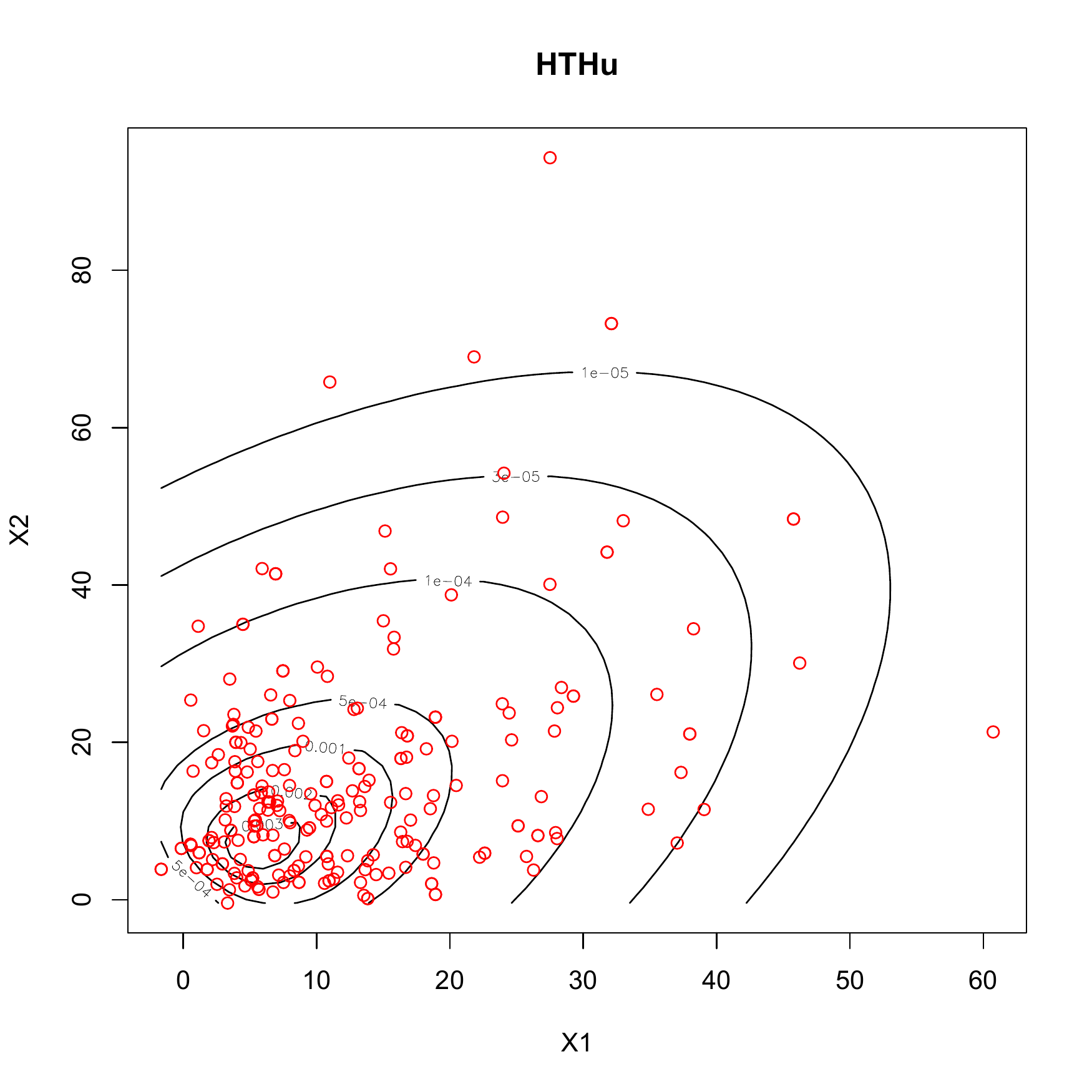} \quad
  \includegraphics[width=0.48\linewidth]{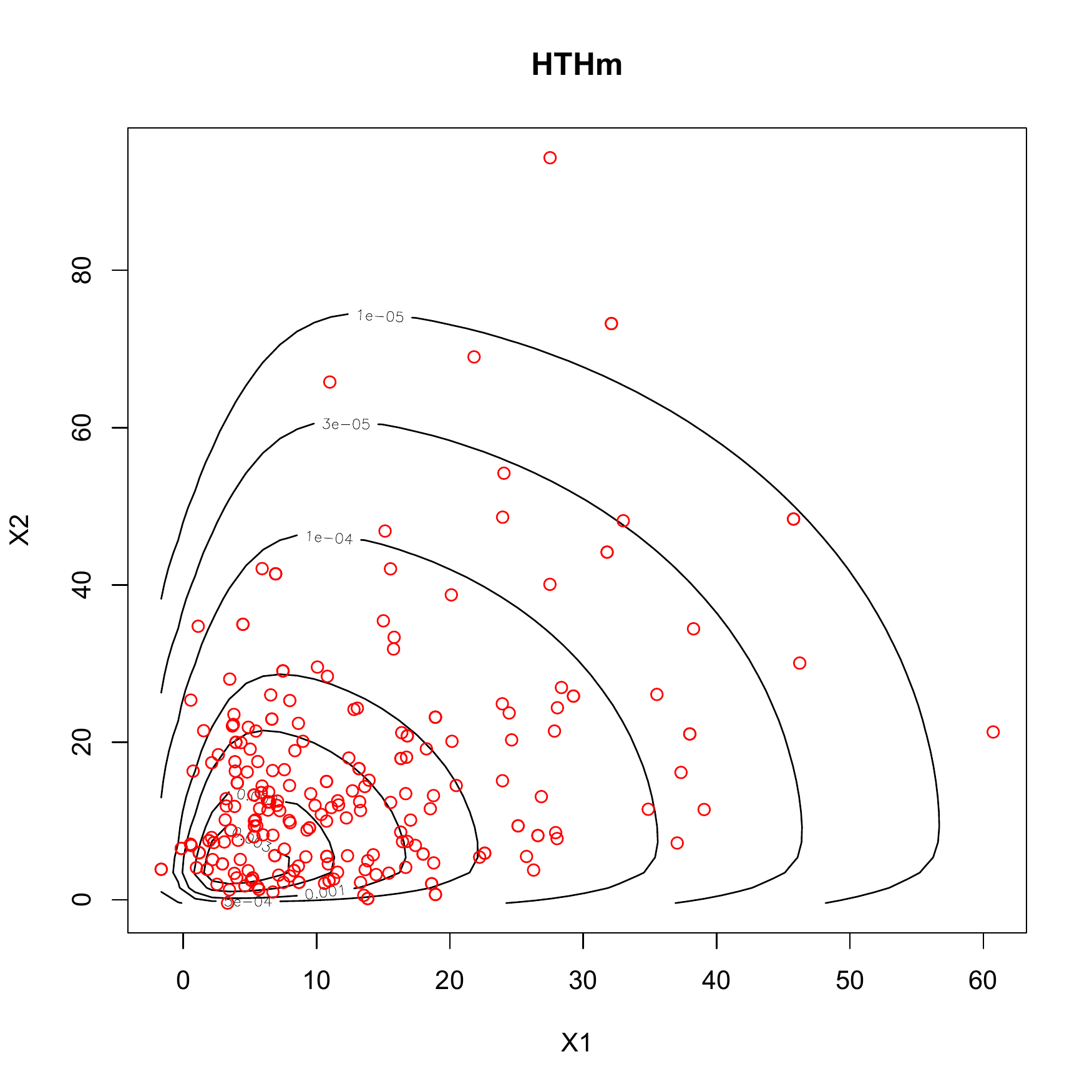}
  \caption{Contour plots for the HTHu and HTHm mixture models on the simulated data with $q=2$.}
  \label{fig:contour2}
\end{figure}

A second simulation study, given in Supplementary Material~S.6, illustrates that the HTHu and HTHm mixtures perform well in terms of parameter estimation and clustering when data are generated from skew-$t$ or HTH distributions. From these simulations, we also see that the HTHu and HTHm mixtures perform comparably to the SDB and classical skew-$t$ mixtures in this regard.

\subsection{Data Analyses}

The performance of our HTH mixture models is assessed on two real data sets. %(Table~\ref{tab:data}). %contains information regarding the dimension, the number of clusters, and source of each data set.
Although use no knowledge of the true class labels in these analyses, they are known and so the adjusted Rand index \citep[ARI;][]{hubert85} can be used to compare the true and predicted classifications. Note that an ARI value of 1 corresponds to perfect agreement between the two sets of labels, the expected value of the ARI under random classification is zero, and a negative ARI value indicates classification that is worse than would be expected by chance.

The hematopoietic stem cell transplant (HSCT) data was collected by the British Columbia Cancer Agency and contains measurements on 9,780 cells obtained via flow cytometry. Note that 78 cells were identified as ``dead" and were removed prior to analysis. In the remaining 9,702 cells, four clusters were detected by manual expert clustering. 
\cite{charytanowicz10} report measurements on 210 kernels from combine-harvested wheat grains collected from experimental fields. The kernels belong to three varieties of wheat: Kama, Rosa, and Canadian. Measurements were collected using a soft X-ray visualization technique to study the internal structure of the kernels. These data are available via the UCI Machine Learning Repository and is hereafter referred to as the seeds data. We consider a subset of three variables, compactness of kernel, length of kernel, and length of kernel groove, and attempt to cluster the kernels based on variety. \cite{lee13c}
used the HSCT and seeds data sets to illustrate the clustering ability of the SDB and classical skew-$t$ mixture approaches. 
%@article{lee13c,
%author="Lee, Sharon X. and McLachlan, Geoffrey J.",
%title="Model-based clustering and classification with non-normal mixture distributions",
%journal="Statistical Methods \& Applications",
%year="2013",
%volume="22",
%number="4",
%pages="427--454",
%}
%
%@misc{tortora15,
%	Note = {{arXiv}:1403.2332},
%	Title = {A Mixture of Coalesced Generalized Hyperbolic Distributions},
%	Author = {C. Tortora and B.C. Franczak and R.P. Browne and P.D. McNicholas},
%	Year = {2015},
%}
%\begin{table}[ht]
%\small
%\caption{Summary of the data sets used in our real data analyses.}
%\begin{tabular*}{1.0\textwidth}{@{\extracolsep{\fill}}lccrr}
%\hline
%Data Set & $G$ & $p$ & Source  \\ 
%\hline
%%Wine&3&13&{\tt gclus} {\sf{R}} package \citep{hurley12}\\
%%Wine&3&27&{\tt pgmm} {\sf{R}} package \citep{mcnicholas11}\\
%%Banknote& 2&6&{\tt mclust} {\sf{R}} package \citep{fraley14}\\
%% Bankruptcy&2&2&{\tt MixGHD} {\sf{R}} package \citep{tortora14} \\
%HSCT&4&4&\cite{lee13} \\
%%AIS& 2&11& {\tt alr3} {\sf{R}} package \citep{weisberg11} \\
%%AIS 2D &2&2& {\tt alr3 }{\sf{R}} package - subset of AIS data set \citep{weisberg11}  \\
%Seeds &3&3&UCI Machine Learning Repository \citep{charytanowicz10} \\
%%Iris&3&4&UCI Machine Learning Repository \citep{fisher36}\\
%\hline
%\end{tabular*}
%\label{tab:data}
%\end{table}

We compare our results to those obtained by fitting a classical skew-$t$ mixture model and an SDB skew-$t$ mixture model using the {\tt EMMIXskew} \citep{wang13} and {\tt EMMIXuskew}  \citep{lee13a} packages, respectively, for {\sf R}. In all cases, $k$-means starts are used {\color{black} and data are scaled prior to analysis}. %In Table \ref{tab:data}, FM-HTHu corresponds to a HTHu mixture model ($q=1$) and the results given under FM-HTHm correspond to the HTHm case $(1<q\leq p)$. We report the result for the best value of $q\in(1,p]$ as selected by the BIC. 
%
%In our analyses, we fit all models for the correct value of $G$. In a true clustering problem, the number of subpopulations is unknown and we would use some model selection criterion such as the BIC to select the most likely number of groups. However, 
The purpose of these analyses is to compare the clustering ability of the models introduced herein to that of the classical and SDB skew-$t$ mixture models; the SDB skew-$t$ mixture model is regarded by some as the state of the art approach (see \cite{lee13c}). The best way to do a direct comparison of clustering ability is to take the issue of selection of the number of components ``out of the equation" so to speak and so we $G$ equal to the number of classes in these analyses. Where applicable, the BIC is used to select $q$. 

The results (Table~\ref{tab:paramEst}) show that the HTHu and HTHm mixture models outperform both the classical and SDB skew-$t$ mixtures for both data sets. This is crucial when one considers that the HSCT and seeds data sets were used by \cite{lee13c} to illustrate the excellent clustering performance of the SDB and classical skew-$t$ mixture approaches. For both data sets, the the HTHu and HTHm mixture models give similarly excellent clustering performance. Notably, the classical skew-$t$ mixture approach gives much better clustering performance than the SDB skew-$t$ mixtures for the seeds data; however, the SDB skew-$t$ mixtures performs better than the classical skew-$t$ mixture approach for the HSCT data. This supports the view of \cite{azzalini14} that neither one of these skew-$t$ formulations should be considered superior to the other. %; this is contrary to the position of \cite{lee13} and others who present results that clearly favour the SDB skew-$t$ mixture approach.
%We see that for both data sets, the HTHu and HTHm mixture models outperform the FM-rMST and FM-rMST models.
%We do this in the same fashion as clustering analyses done by many other authors including \cite{baek11} and \cite{mcnicholas10d}.
%%
\begin{table}[ht]
\caption{ARI values for the four mixture models fitted to the HSCT and seeds data sets.}
\centering
\begin{tabular*}{1.0\textwidth}{@{\extracolsep{\fill}}lrrrr}
\hline
 & HTHu mix.& HTHm mix.& Classical skew-$t$ mix.& SDB skew-$t$ mix.\\ 
\hline
%Wine (p=13)&0.455&0.421&0.830&\\
%Wine (p=27)&0.486&DNF&DNF&DNF\\
%Banknote&0.980&0.980&0.980&0.980 \\
% Bankruptcy&0.257&0.323&0.145&0.103\\
HSCT&0.976&0.984&0.782&0.890 \\
%AIS&  0.884&0.903 & 0.903 &0.941\\
%AIS 2D &0.829&0.775&0.810&0.549 \\
Seeds &0.877&0.877&{\color{black}0.836}& 0.009\\
%Olive Oil &0.997 &&0.523&0.498\\
%Iris&0.904&0.904&0.662& 0.941\\
%Crabs and coffee data did poorly
\hline
\end{tabular*}
\label{tab:paramEst}
\end{table}

\normalsize

%\subsubsection{Banknote Data}
%
%We consider the Swiss Banknote data set \citep{flury88} obtained from the {\tt gclus} package for {\sf R}. The data contains 6 measurements on 200 Swiss banknotes: 100 gelambdaine notes and 100 counterfeit. We consider all possible pairs of the 6 variables and fit our HTH mixture model for $G=1$ to the 15 resultant bivariate data sets. We compare the results in terms of ARI from the specific ($q=1$) and generalized ($q=2$) models for the 15 data sets in Figure \ref{fig:bank}. Where a point falls on the line, the specific model achieves an identical ARI value to the generalized model. Where a point falls below the line, the specific model outperforms the generalized model. We see that the models obtained similar results on all data sets, in some instances obtaining identical results. We found similar results on other data sets. Consequently, we have found no evidence as of yet to go beyond the specific case.
%
%\begin{figure}[hbt]
%\centering
%  \includegraphics[width=0.6\linewidth]{Banknote_Results.pdf} 
%  \caption{ARI values from analysis of the banknote data using the specific and generalized HTH mixture models.}
%  \label{fig:bank}
%\end{figure}

\section{Conclusion}

The HTH distribution was introduced.  This distribution models skewness via a $p \times q$ skewness matrix where $p$ is the dimension of the data and $1\leq q \leq p$.  In this way, the HTH distribution encapsulates both HTHu and HTHm forms of the hyperbolic distribution --- as well as some formulations of the skew-$t$ and skew-normal distributions --- as special cases. We proved identifiability, discussed convexity, and derived the first two moments of the truncated symmetric hyperbolic distribution.  

In a true clustering problem it is desirable to model the data using a flexible distribution that approximates other distributions as special or limiting cases. In this way, it is advantageous to avoid making unnecessary and invalid assumptions regarding the distribution of the data. In this paper, we demonstrate excellent clustering results on two real data sets using the HTHu and HTHm mixture approaches, respectively. 
In the fitting of a mixture of HTH distributions, we are required to compute several integrals on each iteration of our algorithm. Naturally, this can be quite computationally burdensome, particularly in high dimensions. Currently, we have implemented our ECM algorithms in serial $\sf{R}$; future work will focus on developing a parallel implementation. Further work will also focus, \textit{inter alia}, on obtaining information-based standard errors for $\load$ --- perhaps in an analogous fashion to \cite{wang16}.

\section*{Acknowledgements}{\small %The authors are grateful to three anonymous reviewers who provided very helpful feedback, comments, and suggestions, as the result of which the paper has notably improved. 
The authors gratefully acknowledge the support of an Ontario Graduate Scholarship (Murray) as well as respective Discovery Grants from the Natural Sciences and Engineering Research Council of Canada (Browne, McNicholas) and the Canada Research Chairs program (McNicholas).}

\appendix

\section{Convexity of the HTH Distribution}
 
From Section \ref{sec:concave}, convexity of the HTH distribution cannot be proven for $-(p+q)/2 < \lambda < (p+q+1)/2$. In this appendix, we give examples of contours generated for $\lambda$ in this range. Note that when there is no skewness, the contours will be elliptical and thus convex. As of yet, we have not be able to generate a situation with  $-(p+q)/2 < \lambda < (p+q+1)/2$ where the contours are not convex.
For example, we consider data generated with $p=2$, $q=1$ and 
\begin{alignat*}{4}
\vecmu & = \begin{pmatrix} 1 \\ 1 \end{pmatrix} & \qquad \vecSigma & = \begin{pmatrix} 1.5 & 0.3 \\ 0.3 & 2 \end{pmatrix} & \qquad  &\omega= 2 & \qquad 
\vecLambda & = \begin{pmatrix} 9 \\ -5 \end{pmatrix}.
\end{alignat*}
We plot the contours of the distribution for $-2<\lambda<2$ (Figure \ref{fig:HTHuConvex}). Similarly, we consider data generated with $p=2$, $q=2$ and 
\begin{alignat*}{4}
\vecmu & = \begin{pmatrix} 1 \\ 1 \end{pmatrix} & \qquad \vecSigma & = \begin{pmatrix} 1.5 & 0.3 \\ 0.3 & 2 \end{pmatrix} & \qquad  &\omega= 2 & \qquad 
\vecLambda & = \begin{pmatrix} -1 & 9 \\ 3 & 9 \end{pmatrix}. 
\end{alignat*}
The contours of the distribution are plotted for $-2<\lambda<2$ (Figure \ref{fig:HTHmConvex}). In both examples that all of the contours are convex. In our experience, this is typical.
\begin{figure}[!ht]
\centering
  \includegraphics[width=\linewidth]{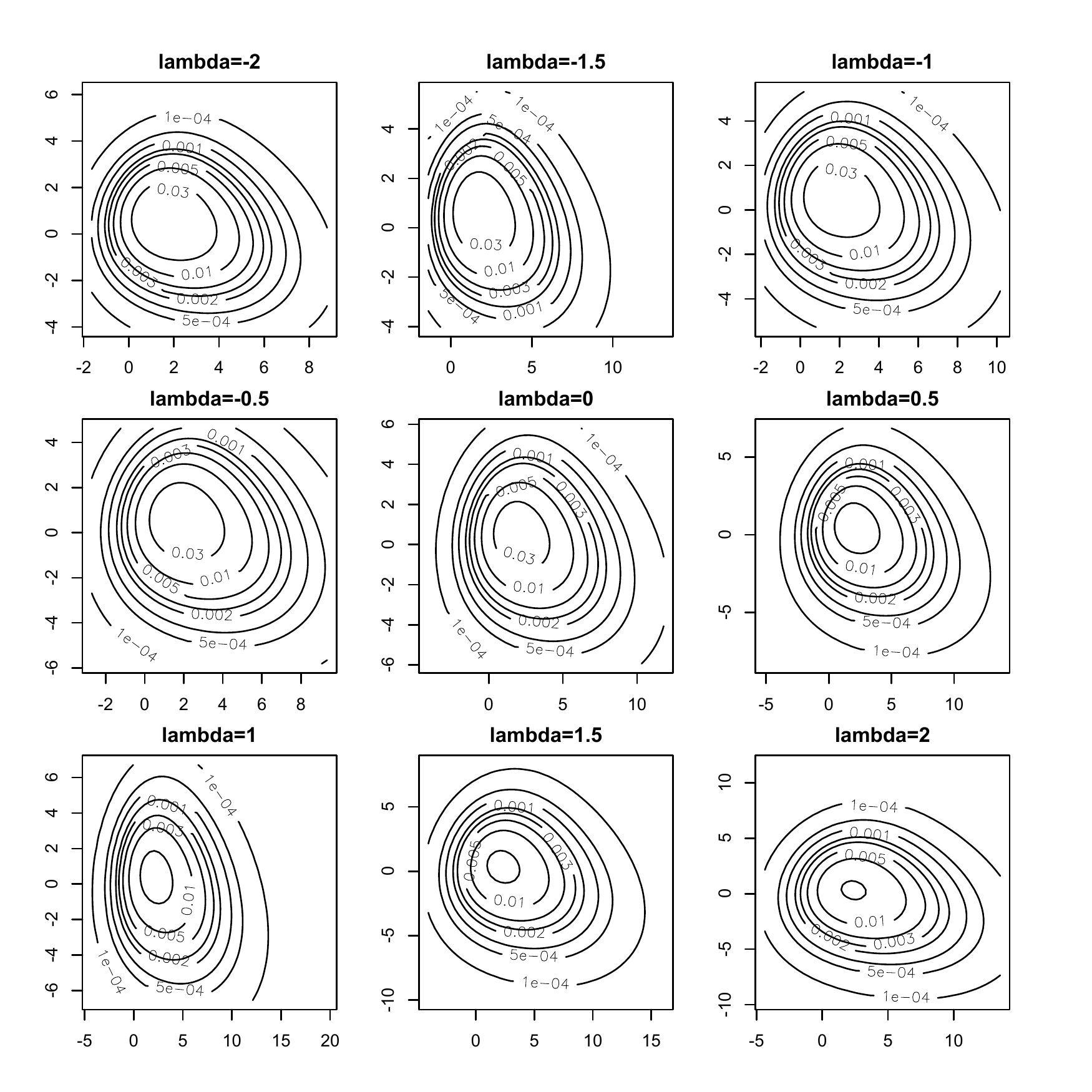} 
  \vspace{-0.3in}
  \caption{Contour plots for the HTHu distribution for several values of $\lambda$.}
  \label{fig:HTHuConvex}
\end{figure}
\begin{figure}[!ht]
\centering
  \includegraphics[width=\linewidth]{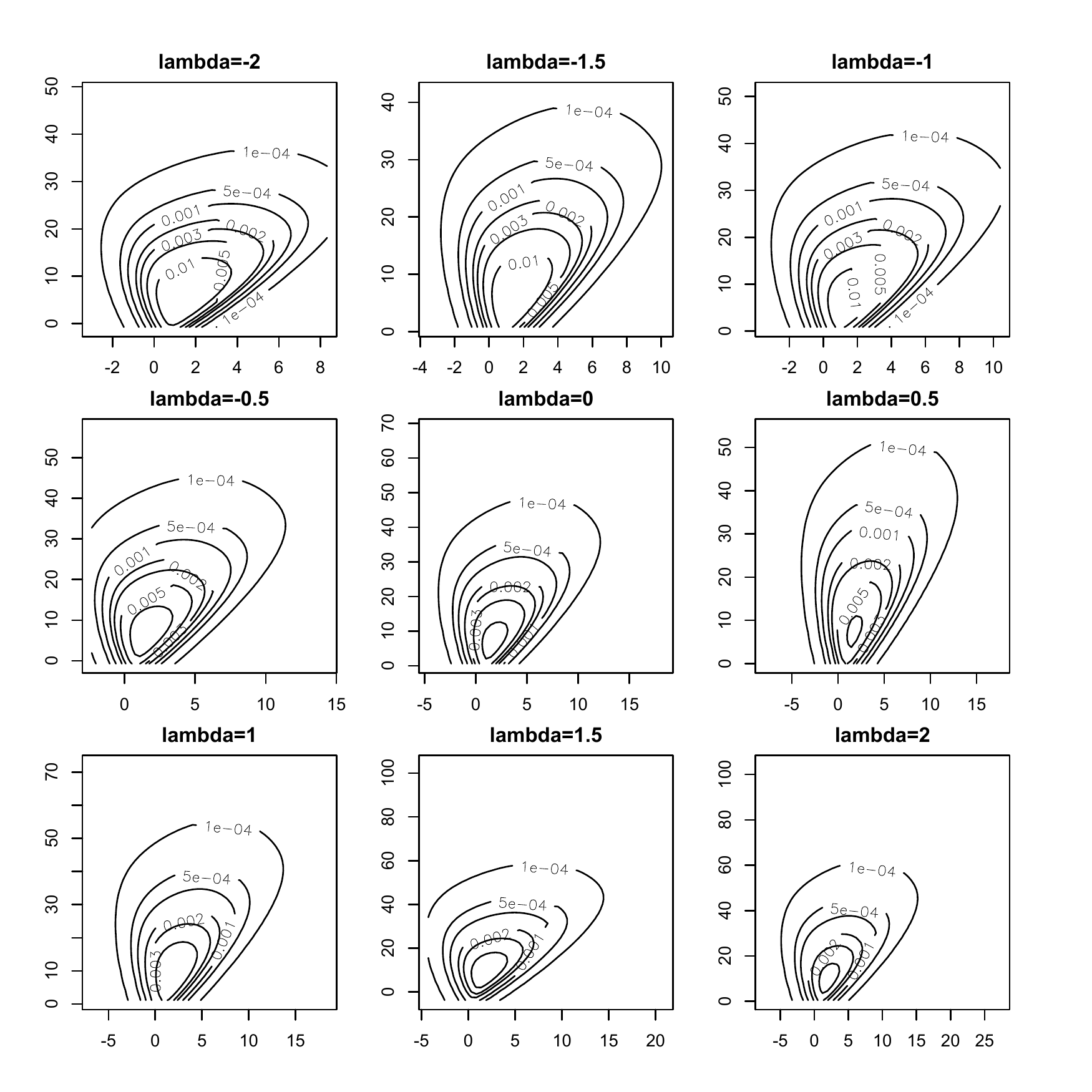} 
  \vspace{-0.3in}
  \caption{Contour plots for the HTHm distribution for several values of $\lambda$.}
  \label{fig:HTHmConvex}
\end{figure}

%\section*{Supplementary Material}
%Proofs, additional mathematical details, and a second simulation study are provided in Supplementary Material. %This file is intended as online supplementary material.

\end{document}